\documentclass[11pt,letterpaper]{article}

\usepackage[margin=1in]{geometry}
\usepackage{latexsym,graphicx,amssymb,amsbsy}
\usepackage{amsthm}
\usepackage{amsmath}
\usepackage{mathdots}
\usepackage{float}
\usepackage{xspace}
\usepackage{paralist}
\usepackage{enumerate}
\usepackage{cases}
\usepackage{caption}

\usepackage[ruled]{algorithm2e} 
\usepackage{algorithmic}

\usepackage{blkarray}

\usepackage{xcolor}
\usepackage{multicol}
\usepackage{graphicx}
\usepackage{subcaption}
\usepackage{varwidth}
\usepackage{wrapfig}
\usepackage{bbm}
\usepackage{tcolorbox}

\newenvironment{proofof}[1]{{\vspace*{5pt} \noindent\bf Proof of #1:  }}{\hfill\rule{2mm}{2mm}\vspace*{5pt}}
\numberwithin{figure}{section}
\numberwithin{equation}{section}
\newtheorem{definition}{Definition}[section]
\newtheorem{remark}{Remark}[section]
\newtheorem{corollary}{Corollary}[section]
\newtheorem{theorem}{Theorem}[section]
\newtheorem{lemma}{Lemma}[section]
\newtheorem{claim}{Claim}[section]

\newtheorem{observation}{Observation}[section]

\newcommand{\be}{\begin{equation}}
\newcommand{\ee}{\end{equation}}
\newcommand{\beq}{\begin{equation*}}
\newcommand{\eeq}{\end{equation*}}

\newcommand{\argmax}{\mathop{\rm argmax}}

\newcommand{\R}{\mathbb{R}}
\newcommand{\N}{\mathbb{N}}

\newcommand{\eps}{\varepsilon}
\newcommand{\ind}[1]{\mathbbm{1}\left[\vphantom{\sum}#1\right]}

\newcommand{\AutoAdjust}[3]{\mathchoice{ \left #1 #2  \right #3}{#1 #2 #3}{#1 #2 #3}{#1 #2 #3} }
\newcommand{\Xcomment}[1]{{}}

\newcommand{\InBrackets}[1]{\AutoAdjust{[}{#1}{]}}
\newcommand{\Ex}[2][]{\operatorname{\mathbf E}_{#1}\InBrackets{#2}}
\newcommand{\Exlong}[2][]{\operatornamewithlimits{\mathbf E}\limits_{#1}\InBrackets{#2}}
\newcommand{\Prx}[2][]{\operatorname{\mathbf{Pr}}_{#1}\InBrackets{#2}}
\newcommand{\Prlong}[2][]{\operatornamewithlimits{\mathbf{Pr}}\limits_{#1}\InBrackets{#2}}

\newcommand{\eqdef}{\overset{\mathrm{def}}{=\mathrel{\mkern-3mu}=}}
\newcommand{\vect}[1]{\ensuremath{\mathbf{#1}}}

\newcommand\restr[2]{{
  \left.\kern-\nulldelimiterspace 
  #1 
  \vphantom{\big|} 
  \right|_{#2} 
  }}
\def\prob{\Prx}

\def\expect{\Ex}

\newcommand{\setmi}[1][i]{\ensuremath{S_{\text{-}#1}}}
\newcommand{\Tsetmi}[1][i]{\ensuremath{T_{\text{-}#1}}}
\newcommand{\Umi}[1][i]{\ensuremath{U_{\text{-}#1}}}

\newcommand{\oset}{\ensuremath{\widetilde{S}}}
\newcommand{\os}{\ensuremath{\widetilde{s}}}

\newcommand{\otset}{\ensuremath{\widetilde{T}}}
\newcommand{\ot}{\ensuremath{\widetilde{t}}}

\newcommand{\rhomi}[1][i]{\ensuremath{\rho_{\text{-}#1}}}

\newcommand{\klog}[1][k]{\ensuremath{\log^{(#1)}}}

\newcommand{\monogaps}{\textbf{Mono-Gaps}}
\newcommand{\mg}{\textsf{MG}}
\newcommand{\expgaps}{\textbf{Exp-Gaps}}
\newcommand{\eg}{\textsf{EG}}

\newcommand{\seqd}{\vect{d}}
\newcommand{\seqg}{\vect{g}}

\newcommand{\numbs}{\ensuremath{\vect{v}}}
\newcommand{\numbts}{\ensuremath{\tilde{\vect{v}}}}
\newcommand{\vt}{\tilde{v}}
\newcommand{\universe}{\ensuremath{\mathcal{U}}}

\newcommand{\sym}{\texttt{Sym}}
\newcommand{\Act}{A}
\newcommand{\Acts}{\mathcal{A}}
\newcommand{\acts}{\mathbf{a}}

\newcommand{\tacts}{\tilde{\acts}}

\newcommand{\Ord}{\texttt{Ord}}
\newcommand{\Card}{\texttt{Card}}

\newcommand{\vals}{\vect{v}}

\newcommand{\alg}{\textsf{ALG}}
\newcommand{\alcard}{\mathcal{A}^{\Card}}
\newcommand{\allev}{\mathcal{A}^{\texttt{Lev}}}
\newcommand{\alord}{\mathcal{A}^{\texttt{Ord}}}
\newcommand{\simul}{\texttt{Sim}}

\newcommand{\distset}{\mathcal{F}}
\newcommand{\distlev}{\mathcal{F}^{\texttt{lev}}}
\newcommand{\distpi}{\mathcal{D_{\pi}}}
\newcommand{\distsigma}{\mathcal{D_{\sigma}}}
\newcommand{\uni}{\text{Uni}}
\newcommand{\supp}{\texttt{supp}}

\newcommand{\dtv}{d_{\textsf{TV}}}

\newcommand{\di}[1][i]{\ensuremath{\textup{M}_{#1}}}
\newcommand{\udi}{\ensuremath{\mathcal{U}}}
\newcommand{\undi}{\ensuremath{\mathcal{V}}}

\newcommand{\lev}{\ensuremath{\rho}}
\newcommand{\levi}[1][i]{\ensuremath{\rho}_{#1}}

\newcommand{\dgap}{\tilde{d}}
\newcommand{\dgapi}[1][i]{\dgap_{#1}}
\newcommand{\dgapv}{\tilde{\vect{d}}}

\definecolor{amaranth}{rgb}{0.9, 0.17, 0.31}

\title{Online Ordinal Problems: Optimality of Comparison-based Algorithms and their Cardinal Complexity}

\author{
Nick Gravin\thanks{ITCS, Shanghai University of Finance and Economics. Email: \texttt{\{nikolai,tang.zhihao\}@mail.shufe.edu.cn}}
\and
Enze Sun\thanks{The University of Hong Kong. Email: \texttt{sunenze@connect.hku.hk}}
\and
Zhihao Gavin Tang\footnotemark[1]
}

\date{}

\begin{document}

\maketitle

\begin{abstract}

We consider ordinal online problems, i.e., tasks that only require pairwise comparisons between elements of the input. A classic example is the secretary problem and the game of googol, as well as its multiple combinatorial extensions such 
as $(J,K)$-secretary, $2$-sided game of googol, ordinal-competitive matroid secretary. A natural approach to these tasks is to use ordinal online algorithms that at each step only consider relative ranking among the arrived 
elements, without looking at the numerical values of the input. We formally study the question of how cardinal algorithms (that can use numerical values of the input) can improve upon ordinal algorithms. 

We give first a universal construction of the input distribution for any ordinal online problem, such that the advantage of any cardinal algorithm over the ordinal algorithms is at most $1+\varepsilon$ for arbitrary small $\varepsilon> 0$. 
This implies that lower bounds from [Buchbinder, Jain, Singh, MOR 2014], [Nuti and Vondr{\'{a}}k, SODA 2023] hold not only against any ordinal algorithm, but also against any online algorithm. 
Another immediate corollary is that cardinal algorithms are no better than ordinal algorithms in the matroid secretary problem with ordinal-competitive objective of [Soto, Turkieltaub, Verdugo, MOR 2021].
However, the value range of the input elements in our construction is huge: $N=O\left(\frac{n^3 \cdot n!\cdot n!}{\varepsilon}\right)\uparrow\uparrow(n-1)$ (tower of exponents) for an input sequence of length $n$. 
As a second result, we identify a class of natural ordinal problems and find cardinal algorithm with a matching advantage of $1+ \Omega \left(\frac{1}{\log^{(c)}N}\right),$ where $\log^{(c)}N=\log\log\ldots\log N$ with $c$ 
iterative logs and $c$ is an arbitrary constant $c\le n-2$. This suggests that for relatively small input numerical values $N$ the cardinal algorithms may be significantly better than the ordinal 
algorithms on the ordinal tasks, which are typically assumed to be almost indistinguishable prior to our work. This observation leads to a natural complexity measure (we dub it cardinal complexity) for any given ordinal online task: the minimum size $N(\varepsilon)$ 
of different numerical values in the input such the advantage of cardinal over ordinal algorithms is at most $1+\varepsilon$ for any given $\varepsilon>0$. As a third result, we show that the game of googol has much lower cardinal complexity of  
$N=O\left(\left(\frac{n}{\varepsilon}\right)^n\right)$.

 \end{abstract}

\section{Introduction}
\label{sec:intro}

The celebrated secretary problem is a key question studied in the optimal stopping theory. 
According to Ferguson~\cite{ferguson1989}, it was first published by Gardner~\cite{Gardner60} in the form of the game of googol:
\begin{tcolorbox}[frame empty]
\paragraph{Game of Googol.}
Ask someone to take as many slips of paper as he pleases, and on each slip write a different positive number. The numbers may range from small fractions of $1$ to a number the size of a googol ($10^{100}$) or even larger. These slips are turned face down and shuffled over the top of a table. One at a time you turn the slips face up. The aim is to stop turning when you come to the number that you guess to be the largest of the series. You cannot go back and pick a previously turned slip. If you turn over all the slips, then of course you must pick the last one turned.
\end{tcolorbox}
A more popular and broadly known version of the game is the secretary problem. Here, one observes a sequence of candidates arriving in a random order and wants to hire the best secretary. The difference with the game of googol is that the online algorithm does not see the numerical values of the candidates, and can only do pairwise comparisons between them, while in the game of googol, the algorithm can use \emph{cardinal} (numerical) values. It is well known that the optimal strategy of the ordinal variant can select the best candidate with probability $1/e$, see, e.g., \cite{Dynkin1963TheOC}.

The secretary problem was a precursor of what is now often referred to as \emph{random arrival} models (sometimes they are also called Secretary models) in many online combinatorial optimization scenarios such as multiple secretaries~\cite{Kleinberg05}, online matching~\cite{KesselheimRTV13,Reiffenhauser19,HoeferK17,EzraFGT22}, network design~\cite{KorulaP09}, selection of the basis in a matroid~\cite{BabaioffIKK18,mor/SotoTV21,Lachish14,FeldmanSZ18}, etc. 
Typically, the objective in those multi-choice combinatorial problems is to maximize the sum of values for the selected feasible subset of items, also known as the social welfare in economics applications. We call it a \emph{cardinal objective} as it depends on the numerical values of the items. In contrast, the objective of the game of googol is \emph{ordinal} as the task of selecting the largest number can be defined merely using pairwise comparisons. 

A number of papers (e.g.,~\cite{AbrahamIKM07,AnshelevichP16,AnshelevichS16,ChakrabartyS14,HoeferK17}) study ordinal (comparison-based) algorithms as an approximation to the offline optimum in different combinatorial problems with the cardinal objective. For instance, Hoefer and Kodric~\cite{HoeferK17} studied ordinal online algorithms in a large variety of secretary models. These algorithms are more practical as they avoid a potentially demanding task of precise value estimations and often are conceptually simpler than their cardinal counterparts.
Another direction is to change the cardinal objective with its ordinal relaxation. Soto et al.~\cite{mor/SotoTV21} considered such relaxations as ordinal-competitiveness for attacking the notorious matroid secretary problem. Buchbinder et al.~\cite{mor/BuchbinderJS14} study the $J$-choice, $K$-best secretary problem to get an approximate results for the cardinal problem of choosing a subset with the maximum sum of top $K$ values. Correa et al.~\cite{JMLR/CorreaCES22}, and Nuti and Vondr{\'{a}}k~\cite{soda/NutiV23} studied the 2-sided game of googol that is motivated by the prophet secretary problem with samples.

A significant effort in the aforementioned work has been directed to lower bounds (hardness of approximation results). 
Unfortunately, there have only been sporadic successes on this front. To the best of our knowledge, all tight lower bound results for cardinal objectives rely on special families of cardinal instances that are essentially ordinal tasks. E.g., the lower bound of $1/e$ for online matching in bipartite graphs by Kesselheim et al.~\cite{KesselheimRTV13} relies on the classic secretary lower bound for the ordinal task of selecting the maximum, while Ezra et al.~\cite{EzraFGT22} derive a tight lower bound for secretary matching in general graphs by analyzing an ordinal task of matching top two vertices in a vertex-weighted graph. 
A possible explanation for the limited progress on the complexity front is that general cardinal algorithms are too difficult to describe and analyze, especially in the multi-choice combinatorial settings.

Even for the much better behaved ordinal objectives, the hardness of approximation results are usually derived against ordinal algorithms. 
For example, despite that the $(J,K)$-secretary problem and the 2-sided game of googol have ordinal objectives, Buchbinder, Jain, and Singh~\cite{mor/BuchbinderJS14} only proved the optimality of their algorithm among ordinal algorithms for $(J,K)$-secretary, and Nuti and Vondr{\'{a}}k~\cite{soda/NutiV23} established a $0.5024$ hardness result only with respect to ordinal algorithms for the two-sided game of googol.  
While it is widely believed that there is no gap between cardinal and ordinal algorithms on ordinal problems, we do not have a formal proof of this fact yet.  Specifically, the original 
paper of Gardner~\cite{Gardner60} made an appealing but informal argument that for large enough numbers the cardinal and ordinal problems are the same.
This intuition was confirmed $30$ years after Martin Gardner's paper by Ferguson~\cite{ferguson1989}, but only for the basic problem of the game of googol. 
He proved that the difference between winning probabilities in the cardinal and ordinal variants is at most $\eps$ for arbitrary small $\eps>0$ when the game is played over sufficiently large integers. Gnedin~\cite{Gnedin94} further showed that the difference completely vanishes when the values can be real numbers. 
One would naturally expect Gardner's intuition to generalize to any ordinal problem. Indeed, if an ordinal objective does not depend on the actual values but only on their relative ranking, it seems obvious that cardinal algorithms should not do better than ordinal. 
This question however is much deeper than it appears at a first glance.

We illustrate the challenge of obtaining good lower bounds on the very well known offline computational task of sorting integers.
There is an endless list of existing sorting algorithms such as bubble-sort, heapsort, quicksort, etc. The vast majority of them are \emph{ordinal} algorithms, i.e., they only do pairwise comparisons between the input elements.  
It is also well known that any such algorithm has a fundamental limitation: on average, it must perform at least $\Omega(n\log n)$ comparisons to produce the correct output.
On the other hand, there are a few algorithms such as pigeonhole, counting, and radix sorts that utilize the cardinal information about the input. I.e., these algorithms are not comparison based and thus are not limited by the $\Omega(n\log n)$ barrier. 
Some of them have faster than $O(n\log n)$ running time for the practical task of  sorting integers in a limited range from $0$ to $N$, see, e.g., $O(n\sqrt{\log\log N})$ randomized algorithm  of \cite{HanT02}, or deterministic $O(n \log\log N)$ algorithm of \cite{Han04} in the word RAM model of computations. There is no known tight lower bound for the problem of integer sorting, and it is unlikely that we will see such a bound any time soon.

The story of the sorting algorithms illustrates how cardinal information may be advantageous in performing \emph{ordinal tasks}, i.e., problems whose outputs only depend on the pairwise comparisons between the elements of the input.
In this paper, we study what advantage one can get by using the cardinal information in ordinal tasks, but instead of computational problems (which can be tricky to formalize due to the differences between many models of computations) we consider online problems with the focus on the information theoretic guarantees.

\subsection{Model: Online Ordinal Problems}
\label{sec:model}
In order to discuss our contributions accurately, we first formalize the class of online ordinal problems that captures a few variants of the secretary problem. We focus on online maximization problems throughout the paper. 

Let $\mathcal{U}=\{e_1,e_2,\dots,e_n\}$ be the universe of $n$ elements.
The elements arrive one-by-one in a random order $\pi = \left(\pi(1),\pi(2), \dots, \pi(n)\right) \in \sym(n)$ drawn from a priori known distribution $\distpi$. We use $\pi[k]$ to denote the first $k$ arrivals $(\pi(1),\pi(2),\dots,\pi(k))$.
The adversary specifies $n$ \emph{distinct} integers $\vals = (v_1,v_2,\dots,v_n) \in [N]^n$ for the $n$ elements of $\universe$. 
We will also refer to $\vals$ as a set-permutation pair $(S,\sigma)$, where $S=\{v_1,\ldots,v_n\} \subseteq [N]$ is an unordered set of all numbers in $\numbs$, and $\sigma \in \sym(n)$ is their ranking.
That is, $v_i$ corresponds to the $\sigma(i)$-th largest number in the set $S$. We will write $(S(\numbs), \sigma(\numbs))$ to denote the corresponding set and the ranking for the vector $\numbs$. 

At each step $k \in [n]$, the element $e_{\pi(k)}$ and its associated number $v_{\pi(k)}$ are revealed. The online algorithm observes identities $\pi[k]$ of the first $k$ elements and the corresponding $k$ numbers $\numbs_{\pi[k]}=(v_{\pi(1)},\dots,v_{\pi(k)})$, and selects an action $a_k=\alg_k(\pi[k],\numbs_{\pi[k]})$ from the action set $\Act_k$. The final output of the algorithm after step $n$ is $\acts=(a_1,\ldots,a_n)$. There could be some constraints on the feasible actions: $\Acts\subseteq\Act_1\times\Act_2\times\ldots\times\Act_n$. We think of the algorithm $\alg(\numbs,\pi)$ as a function $\alg: [N]^n\times\sym(n)\to\Acts$.

\paragraph{Ordinal Reward Functions.}
We study \emph{ordinal reward functions} $R(\acts,\sigma,\pi):\Acts\times\sym(n) \times \sym(n) \to\R_{+}$ and assume that if $\acts\notin\Acts$, then $R(\acts,\sigma,\pi)=0$. We refer to such a setting as ordinal problems since the reward function is determined by 1) the actions $\acts$ taken by the algorithm, 2) the relative order $\sigma$ of the numbers, 3) the arrival order $\pi$ of the elements; and is independent of the actual values $S$.
Then the performance of an algorithm is $\Ex[\pi\sim\distpi]{R(\acts(\numbs,\pi),\sigma(\numbs),\pi)}$.

\paragraph{Ordinal (Comparison-based) Algorithms.}
We study a subfamily of the online algorithms that only use pairwise comparisons to determine which actions to take at every step. 
Formally, an ordinal algorithm takes action $a_k=\alg_k(\pi[k], \sigma(\numbs_{\pi[k]}))$ at step $k$, where $\sigma(\numbs_{\pi[k]})$ is the ranking of the $k$ arrived elements that only depends on the ordinal comparisons of the elements in $\numbs_{\pi[k]}$.
We think of the algorithm as $\alg: \sym(n) \times \sym(n)\to\Acts$, i.e. $\alg(\sigma(\numbs),\pi)$.
We use $\Ord$ to denote the family of all ordinal algorithms and use $\Card$ to denote the family of all algorithms.

\paragraph{Remark.} For notation simplicity, we only formally define deterministic algorithms and notice that any randomized algorithm can be viewed as a mixture of deterministic algorithms.
We stress, however, that our results below hold against randomized algorithms, as our construction does not depend on a specific deterministic and/or randomize algorithm. We discuss the differences between randomized and deterministic algorithms in more detail in Section~\ref{sec:results}.

\subsubsection{Examples}
We give a few examples of ordinal tasks from the literature on random arrival models and show how they fit into our unified model.

\paragraph{Game of Googol.} 
The universe $\universe$ corresponds to the $n$ slips and $\numbs$ are the numbers written on the slips (distinct integers from $1$ to $N=10^{100}$). The arrival order $\pi\sim\distpi$ is drawn uniformly at random.
At each step $k$, the algorithm observes number $v_{\pi(k)}$ and gets two options $A_k = \{\text{accept,reject}\}$. Only one ``accept'' is allowed per the feasibility constraint $\Acts$. The reward function $R$ is $1$ whenever we accept the largest number in $\numbs$ and is $0$ otherwise.

\paragraph{Two-sided Game of Googol.} The game is first introduced by \cite{JMLR/CorreaCES22}, motivated by the prophet secretary problem with samples. It is similar to the game of googol with the following differences.
\begin{enumerate}
	\item The universe $\universe$ consists of $2n$ faces of $n$ slips with $2n$ numbers $\numbs$ written on either side of every card: $e_i$ and $e_{i+n}$ are the two sides of $i$-th card.
	\item Each slip faces up or down with half \& half probability and the $n$ slips are shuffled uniformly at random. I.e., $\pi(k)$ and $\pi(k+n)$ are the two faces of the $k$-th card with $\pi(k)= \rho(k) + x_k$ and $\pi(k+n) = \rho(k) + n-x_k,$ where $\rho \in \sym(n)$ is drawn uniformly at random and $x_k \in\{0,n\}$ with half \& half probability.
	\item The player sees all the face-up numbers. I.e., the action space is empty  for $k \le n$.
	\item The player turns the slips one at a time and aims to stop at the slip with the largest (initially) face-down number. I.e., at step $n+1 \le k \le 2n$, the algorithm observes $v_{\pi(k)}$ and has two options $A_k = \{\text{accept,reject}\}$. At most one ``accept'' is allowed per feasibility constraint $\Acts$. The reward function $R$ is $1$ if we accept the largest in $\{v_{\pi(n+1)},v_{\pi(n+2)}, \dots, v_{\pi(2n)}\}$, and is $0$ otherwise.
\end{enumerate}
Correa et al.~\cite{JMLR/CorreaCES22} established an ordinal algorithm with winning probability $0.4529$ and is recently improved to $0.5009$ by Nuti and Vondr{\'{a}}k~\cite{soda/NutiV23}. \cite{soda/NutiV23} also established a hardness bound of $0.5024$ for ordinal algorithms, as an implication of the results from \cite{campbell1981choosing, soda/CorreaCFOT21,ec/DuttingLLV21}.

\paragraph{J-choice, K-best Secretary.} $(J,K)$-secretary is a generalization of the classical secretary problem, studied by \cite{mor/BuchbinderJS14}. The algorithm is allowed to accept at most $J$ elements and the goal is to select as many as possible from the $K$ largest numbers. According to this definition, $(1,1)$-secretary is the classical secretary problem / the game of googol. It is straightforward to verify that $(J,K)$-secretary is ordinal. Buchbinder, Jain, and Singh~\cite{mor/BuchbinderJS14} derived the optimal ordinal algorithm via linear programming techniques.

\paragraph{(Ordinal) Matroid Secretary.} The matroid secretary problem is first introduced by \cite{BabaioffIKK18} and it remains an intriguing open question whether a constant competitive algorithm exists. The goal is to maximize 
the total sum of values among selected elements, which is not an ordinal objective. We describe below a stronger objective of \emph{ordinal competitiveness} introduced by Soto, Turkieltaub, and Verdugo~\cite{mor/SotoTV21}.
The universe $\universe$ is the ground set of a known matorid and $\numbs$ corresponds to the values assigned to $\universe$. The distribution $\distpi$ is a uniform random order. 
We observe the element $e_{\pi(k)}$ and its value $v_{\pi(k)}$ at step $k$ and get two options $A_k = $\{accept,reject\}. We may accept multiple elements, but under a feasibility constraint $\Acts$ that accepted set is an independent set of the matroid.
Let $\mathsf{OPT}_{i}$ denote the maximum value independent set of $i$ elements for each $1\le i\le r$, where $r$ is the rank of the matroid. One can obtain $\mathsf{OPT}_i$ by greedily selecting $i$ elements using only \emph{pairwise comparisons} between $\numbs$. Thus, the following family of reward functions are ordinal: 
$
R_i = \left| \{ \text{accepted elements}\} \cap \mathsf{OPT}_i \right|, i\in[r].
$
An online algorithm is $\Gamma$-ordinal competitive if and only if $\Ex{R_i} \ge \Gamma \cdot i$ for every $1\le i \le r$.
Soto, Turkieltaub, and Verdugo~\cite{mor/SotoTV21} designed an $O(\log\log r)$-ordinal competitive algorithm for the matorid secretary problem.

\subsection{Our Contributions}
\label{sec:results}
Let us now study online ordinal problems from the perspective of proving lower bounds (impossibility results). 
This task is much easier if the algorithm is restricted to be ordinal. Indeed,
\begin{observation}
\label{observation}
For an arbitrary online ordinal problem, to prove lower bounds against ordinal algorithms, it suffices for the adversary to design the (distribution of) permutations $\sigma$, rather than the (distribution of) values $\numbs=(S,\sigma)$.
\end{observation}

Furthermore, as discussed in the introduction, Martin Gardner~\cite{Gardner60} stated an informal but intuitive argument that cardinal algorithms do not have any advantage over ordinal algorithms, if arbitrarily large numbers are allowed for the set $S$. 
In other words, the Observation~\ref{observation} generalizes to all (even cardinal) algorithms on ordinal tasks. Our first result confirms this intuition. 

\begin{theorem}
\label{thm:ordinal}
For an arbitrary $n$-round online ordinal task with reward function $R:\Acts\times \sym(n) \times \sym(n)\to \R_{+}$, distribution $\distpi$ of orders $\pi$, distribution $\distsigma$ of permutations $\sigma$, and any $\eps>0$, there exists a sufficiently large integer $N\in\N$ and a distribution $\distset$ of sets $S \subseteq [N]$ such that 
\[
\max\limits_{\alg\in\Card} \Exlong[\substack{\pi \sim \distpi \\ \numbs=(S,\sigma) \sim\distsigma \times\distset}]{R(\alg( \numbs, \pi), \sigma, \pi)}\le
(1+\eps)\cdot \max\limits_{\alg\in\Ord}\Exlong[\substack{\pi\sim\distpi \\ \sigma \sim\distsigma}]{R(\alg(\sigma,\pi),\sigma,\pi)}
\]
\end{theorem}
The above theorem states that for an arbitrary online ordinal problem, an arbitrary distribution of the permutation $\sigma$, there exists a distribution of the set $S$ so that cardinal algorithms have at most $\eps$ advantage over ordinal algorithms. 
In other words, any lower bound against ordinal algorithms also works against cardinal algorithms. 
Equivalently,  we confirm the optimality of ordinal algorithms for online ordinal problems.

As immediate implications, we provide an alternative proof of the optimality of the $\frac{1}{e}$ algorithm for the game of googol; strengthen the previous lower bounds of $(J,K)$-secretary, 2-sided game of googol; and conclude that it is without loss of generality to study ordinal algorithms for matroid secretary if the objective is to maximize the ordinal competitive ratio.
\begin{corollary}
The following results hold within all (even cardinal) algorithms: 
\begin{description}
\item[Secretary:] the algorithm of \cite{Gardner60} is optimal.
\item[$(J,K)$-secretary:] the algorithm of \cite{mor/BuchbinderJS14} is optimal.
\item[$2$-sided game of googol:] no algorithm can win with probability larger than $0.5024$.
\item[Matroid secretary:] there is an ordinal algorithm achieving the best ordinal competitive ratio.
\end{description}
\end{corollary}

\paragraph{Comparison with Ramsey style constructions.} A similar question to our Theorem~\ref{thm:ordinal} was considered in the mid $80$s by Moran, Snir, and Manber~\cite{jacm/MoranSM85} about the difference between cardinal and ordinal non-online algorithms for ordinal tasks. They showed, using Ramsey theorem argument, that for any \emph{deterministic} cardinal algorithm with a finite set of possible outputs and size $n$ input vector with the values in a universe $[N]$ of sufficiently large size, there is a subset of $S\subset [N]$ of size $|S|=n$ on which this cardinal algorithm behaves exactly like an ordinal algorithm. I.e., they find an input to a fixed \emph{deterministic} cardinal algorithm on which this algorithm does no better than the best ordinal algorithm. In other words, there is a specific cardinal input $S\subset[N]$ with $|S|=n$ (a response of the adversary to a specific strategy of the algorithm player), such that the $\eps$ from Theorem~\ref{thm:ordinal} is $\eps=0$. 

Much later and independently from~\cite{jacm/MoranSM85} similar Ramsey type arguments were used in two specific 
online problems~\cite{mor/CorreaDFS22, EzraFGT22}. It should be noted, however, that the results of~\cite{jacm/MoranSM85}
cannot be directly applied to those online settings. Indeed, the \emph{randomized} cardinal algorithms (strategies of the algorithm player for unknown adversarial input) usually have a small advantage $\eps>0$ over the ordinal algorithms, while the construction from \cite{jacm/MoranSM85} 
has advantage $\eps=0$. E.g., consider a simple game of googol with $n=2$ cards and values in $[N]$. It is easy to find a cardinal (randomized) algorithm that guesses the maximum with 
probability at least $\frac{1}{2}+\frac{1}{2N}$ (see Section $19.3.3$ in~\cite{MathforCS}), while any ordinal algorithm cannot do better than random guessing with the winning probability of $\frac{1}{2}$.

The Ramsey style approaches of~\cite{mor/CorreaDFS22, EzraFGT22} are problem specific, as they take a specific (randomized) cardinal algorithm and after certain discretization and de-randomization steps combined with Ramsey construction from~\cite{jacm/MoranSM85}, and obtain an instance on which this cardinal algorithm has advantage of at most $\eps$ over the best ordinal algorithms. In contrast, our construction from  Theorem~\ref{thm:ordinal} is universal, i.e., it does not care about the specific algorithm or ordinal task, as it directly constructs a distribution of inputs such that cardinal values give almost no extra information about ordinal ranking of the revealed elements.

\subsubsection{Cardinal Complexity}
Our construction of the distribution $\distset$ is universal. I.e., the distribution is independent of the ordinal problem's structure and the distribution $\distpi$ of permutations.
Intuitively, the cardinal algorithms should not have an advantage over the ordinal algorithms if a subset $T\subset S \sim \distset$ of values reveals no extra information about the rankings of $T$ within $S$. 
We formalize this property and call it \emph{order statistics indistinguishable (OSI)} in the sense that an arbitrary collection of order statics of the random set $S$ share almost the same probability distribution. See Section~\ref{sec:set_construction} for the formal definition. We believe that the distribution of OSI sets is of independent interest and may find applications in other related problems.

On the negative side, our universal construction uses quite large numbers. Indeed, the largest number $N$ in our main theorem is  
$O\left(\frac{n^3 \cdot n! \cdot n!}{\eps}\right)\uparrow\uparrow(n-1)$ for a given $\eps$, where $\uparrow\uparrow$ is the Knuth's up-arrow notation for the iterated exponentiation, i.e., $a \uparrow\uparrow b \eqdef \underbrace{a^{a^{\iddots^{a}}}}_{b\text{ times}}$. 
Furthermore, numbers as large as $10^{100}$ is out of practical range in almost any imaginable scenario. E.g., if anyone was to assign numerical scores to candidates, or research papers she would most likely use integer scores less than $100$ and possibly even smaller than $10$. 
Thus, it is natural to ask how cardinal algorithms can perform better than ordinal algorithms when the largest number $N$ is bounded.
This is a similar story to sorting algorithms, as, while $\Omega(n\log n)$ comparisons are necessary in general, the cardinal algorithm for rather large integers can do significantly better. 

To this end, we introduce the \emph{cardinal complexity} of online ordinal problems, i.e., the minimum number of different integers required so that the advantage of cardinal algorithms over ordinal algorithms is no more than $\eps$. See Section~\ref{sec:universal} for the formal definition.

\paragraph{Tight Cardinal Complexity: Die Guessing.} Our second, and perhaps the most surprising result is that the universal construction is almost optimal regarding the dependency on $\eps$ for general online ordinal problems. Specifically, we prove that the tower of $(n-1)$ exponents is necessary. 

We study a one-shot ordinal game called die-guessing. 
Consider a fair die with $n$ faces, e.g., the standard die with $n=6$. Imagine two players playing the following game. The first player secretly writes $n$ distinct integers from $\{1,2,\ldots,N\}$ on each face and then roll the die. The second player sees all faces but one, which is at the bottom. The second player wins if he guesses correctly the rank of the hidden number compared to all visible ones. 

Without seeing the numbers, by guessing any rank between $1$ and $n$, the second player wins with probability $\frac{1}{n}$. This is an ordinal algorithm in our language.
We construct a cardinal algorithm with $\frac{1}{n}\left(1 + \Omega\left(\frac{1}{\klog[c] N} \right)\right)$
probability of guessing correctly for the game by utilizing the cardinal information, where $\klog[c](x) = \underbrace{\log \log \dots \log x}_{c\text{ times}}$ and $c\le n-2$ is any value.
An important implication of this result is that, in general, using the Gardner's intuition might be infeasible in practice. Indeed, the cardinal values with only a doubly exponential dependency on $\frac{1}{\eps}$ may easily get to the order of $2^{2^{100}}$, which are too large to be compared with each other or even stored on a computer. On the positive side, our result suggests that in some cases one can use cardinal information to improve upon performance of the ordinal algorithm if the numerical values are not very big.

\paragraph{Special Tasks: Game of Googol.} Finally, our previous results do not say anything about the specific task of the googol game, i.e., the task of identifying the maximum in a random sequence, which may admit a more efficient cardinal-to-ordinal reduction than is necessary for the die guessing game. We present a much more efficient construction of cardinal complexity $N=O\left(\left(\frac{n}{\eps}\right)^{n}\right)$ for the game of googol such that for any $n\in\N$ the advantage of any cardinal algorithm over the best ordinal algorithm is at most $\eps$. We obtain this construction as a solution to natural variant of the die guessing game adopted to the game of googol. This variant, which we call \emph{maximum guessing}, has the same setup as the die guessing game, but with a different objective to guess correctly whether the hidden face is the maximum among $n$ numbers written on the faces, or if it is not. This result highlights the role of guessing games as natural hardcore problems that capture difficulty of using cardinal algorithms for ordinal problems.
 
\subsection{Related Work}
\label{sec:related}
Our paper is mostly motivated by the extensive study of secretary problem and its variants. 
Besides the results that we have discussed in the introduction, Chan, Chen, and Jiang~\cite{ChanCJ15} focused on the $(2,2)$-secretary with a cardinal objective (i.e. the sum of the weights of the selected items), and proved that the best ordinal algorithm is $0.488$-competitive while a cardinal algorithm can be $0.492$-competitive, which formally separates cardinal algorithms from ordinal algorithms in this multi-choice secretary problem. 
Kesselheim, Kleinberg, and Niazadeh~\cite{stoc/KesselheimKN15} studied the secretary problem with non-uniform arrival orders and investigated the minimum entropy of the arrival order distribution that permits constant probability of winning. Recently, Hajiaghayi et al.~\cite{HajiaghayiKKO22} generalized their results to multi-choice secretary problems.

We are aware of two related prior works that implicitly analyze the advantage of cardinal algorithms over ordinal algorithms to obtain results in their cardinal models.
First, Correa et al.~\cite{mor/CorreaDFS22} consider the setting of unknown i.i.d.~prophet inequality, proving among other results that no online algorithm has competitive ratio better than $\frac{1}{e}$. Note that the $\frac{1}{e}$ ratio can be achieved by the standard ordinal algorithm for the classic secretary problem despite the fact that the objective is cardinal.
Second, Erza et al.~\cite{EzraFGT22} study the secretary matching setting. They introduce an ordinal version of the problem to establish a tight lower bound of $\frac{5}{12}$. Their ordinal version is a multi-choice secretary setting with the objective to select the maximum element. 

Both papers among other things (i) analyse settings with the goal of selecting the maximum element; (ii) apply a nontrivial Ramsey theory argument~\cite{jacm/MoranSM85} to reduce what we call ``cardinal'' algorithms (i.e., algorithms that observe numerical values) to what we call ``ordinal'' algorithms (i.e., algorithms that only use relative ranking of the elements).
In fact, Erza et al.~\cite{EzraFGT22} explicitly do a two step reduction from their original setting with cardinal objective: first to the ``Hybrid setting'' which is exactly captured by our notion of an ordinal objective; then to the ``Ordinal setting'' where not only the objective but also the algorithm are ordinal. The latter step of their reduction is much more difficult than the former one and was inspired by the Ramsey theory argument from Correa et al.~\cite{mor/CorreaDFS22}. Our universal construction can be used as an alternative proof for the reduction from the hybrid to the ordinal setting. Interestingly, given the connection between i.i.d.~prophet inequality and the secretary settings,
the approach of Correa et al.~\cite{mor/CorreaDFS22} can be almost verbatim applied to the game of googol and the size of their construction\footnote{They only give existential result and understandably did not explicitly calculate its size.} is similar to our universal bound in Section~\ref{sec:universal}.

When proving hardness of approximation results for different random arrival models, the most common choice of the elements ranking $\sigma$ (not to be confused with the arrival order $\pi$) is a uniform distribution over all permutations. 
It gives the optimal lower bound of $\frac{1}{e}$ for the game of googol, and the state-of-the-art lower bound $\approx0.5024$~\cite{soda/NutiV23} for the two-sided game of googol. 
In the context of combinatorial random arrivals models, such choice of $\sigma$ received a name of \emph{random assignment model}. Interestingly, it is not able to rule out the matroid secretary conjecture of Babaioff et al.~\cite{BabaioffIKK18} even for 
order-competitiveness and ordinal algorithms, as there is a $\frac{2e^2}{e-1}$-competitive algorithm of Soto~\cite{sicomp/Soto13} in the random assignment model.

\subsection{Road Map} 
Section~\ref{sec:set_construction} presents the construction of order statistics indistinguishable (OSI) sets. Section~\ref{sec:universal} is dedicated to the proof of our main theorem, as an application of OSI sets. Section~\ref{sec:guessing} shows that the cardinal complexity of the universal construction is essentially tight for the die guessing game. Section~\ref{sec:max} gives a much more efficient construction than the universal construction for the game of googol with only exponential in $n$ and $\eps$ cardinal complexity. 
We conclude with a list of open problems in Section~\ref{sec:open}. More tedious and long proofs are deferred to the Appendices~\ref{app:dtv},\ref{app:guessing},\ref{app:googol}.

\section{Order Statistics Indistinguishable Sets}
\label{sec:set_construction}
In this section we construct an Order Statistic Indistinguishable (OSI) distribution $\distset$ of sets $S\subset[N]$ with $|S|=n$. Before formally defining the OSI property and presenting the construction we introduce a few useful notations that we will use
 throughout the paper. We use $\numbs_I$ to represent a vector $\numbs$ restricted to an index set $I\subseteq [n]$. Similarly, given a set $S = \{s_1, s_2, \cdots,s_n\}$ listed in ascending order $s_1 < s_2 <..<s_n$, $S_I$ denotes 
 the subset $\{ s_k \mid k \in I \}$ for an arbitrary index set $I \subseteq [n]$. We shall also use $\setmi$ to denote the subset $ \{s_k \mid k \ne i \}$. 

Intuitively, cardinal algorithms should not be much better than ordinal algorithms if observing the numbers in a set $S_I$ of $S\sim\distset$ reveals almost no information about the index set $I\subset[n]$. 
That is exactly the OSI property which we would like to achieve. Before we formally state the OSI property presented in Lemma~\ref{lem:core_construction} we will recall the definition and a few useful properties of 
the Total Variation (TV) statistical distance. 
 
 \subsection{Total Variation Distance}
Throughout the paper, we shall study discrete random objects, including integers and ordered sets of integers. Consider two random objects $X,Y$ sampled from probability mass functions $\vect{p_X},\vect{p_Y}$ over a discrete domain $\mathcal{T}$. The total variation distance between random variables $X,Y$ is defined as the following.
\[
\dtv(X,Y) \eqdef\dtv(\vect{p_X},\vect{p_Y}) \eqdef \frac{1}{2} \cdot \sum_{t \in \mathcal{T}} |p_X(t) - p_Y(t)|
\]
The following lemmas summarize certain standard and useful properties of the TV-distance, which we state here for the ease of reference without proofs.
\begin{lemma}[Triangle Inequality]
\label{lem:dtv_tri}
Let $X,Y,Z$ be random objects over a discrete domain $\mathcal{T}$, then $\dtv(X,Z) \le \dtv(X,Y)+\dtv(Y,Z)$.	
\end{lemma}

\begin{lemma}[Mapping]
\label{lem:dtv_map}
Let $X,Y$ be random objects over a discrete domain $\mathcal{T}$ and $f$ be an arbitrary (random) mapping from $\mathcal{T} \to \mathcal{U}$. Then $\dtv(f(X),f(Y)) \le \dtv(X,Y)$.	
\end{lemma}

We prove the following bound on total variation distance of uniform distributions in Appendix~\ref{app:dtv}.
\begin{lemma}[Uniform Distributions]
\label{lem:dtv_uni}
Suppose $x_1 \sim \uni[\alpha_1,\beta_1]$ and $x_2 \sim \uni[\alpha_2, \beta_2]$ with positive integers $0 \le \alpha_2 \le \beta_2 \le \beta_1-\alpha_1$, then $\dtv(x_1, x_1+x_2) \le \frac{\beta_2}{\beta_1-\alpha_1+1}$.
\end{lemma}
 
\subsection{Construction of Order Statistic Indistinguishable Sets.}

\begin{lemma}
\label{lem:core_construction}
For any $\eps>0$ and $N=O\left(\frac{n^2}{\eps}\right)\uparrow\uparrow (n-1)$, there exists a distribution $\distset_n(\eps)$ over $n$-element sets $S \subseteq [N]$ such that
	\[
	\dtv(S_I,S_J) \le \eps, \quad \forall I,J \subseteq [n], |I|=|J|.\quad\quad\quad\text{(OSI property)}
	\]
\end{lemma}

We refer to such a distribution as order statistics indistinguishable since an arbitrary collection of order statistics of the random set $S$ would share the (almost) same probability distribution.
We start with a weaker version of the above lemma.

\begin{lemma}
\label{lem:main_construction}
For any $\eps>0$ and $N=O\left(\frac{1}{\eps}\right)\uparrow\uparrow (n-1)$, there exists a distribution $\distset_n(\eps)$ over $n$-element sets $S \subseteq [N]$ such that 
\[
\dtv(\setmi, \setmi[j]) \le \eps, \quad \forall i,j\in[n].
\]
\end{lemma}

\begin{proof}
We give an explicit construction of the distribution that satisfies the stated property. 
As a warm up we first describe how to construct such distribution $\distset_n$ for $n=2,3$.
\paragraph{Warm up for $n=2$.} For $N= \Theta(\frac{1}{\epsilon})$, consider a uniform distribution over consecutive numbers $\{i,i+1\}$ for all $1 \le i \le N-1$. Then $\setmi[1]$ is a uniform distribution over $\{ \{1\},\{2\},\ldots,\{N-1\} \}$ and $\setmi[2]$ is a uniform distribution over $\{\{2\},\ldots,\{N-1\},\{N\} \}$. Thus, $\dtv(\setmi[1], \setmi[2]) = \frac{1}{N-1} \le \eps$.

\paragraph{Warm up for $n=3$.} For $N=\left(\frac{1}{\eps}\right)^{\frac{1}{\eps}}$, consider a uniform distribution over $\{i,i+2^{\ell},i+2^{\ell+1} \}$ for all $\ell \le \frac{1}{\eps}$ and all $i$'s as long as $i + 2^{\ell+1} \le N$. 
For an observed set $\{i, i+2^\ell \}$, unless $\ell \in \{1,\frac{1}{\eps}\}$ or $i \le 2^{\ell}$, or $i+2^{\ell+1}>N$, it is equally likely that the observed set was obtained after deleting
 $i-2^{\ell}$, or $i + 2^{\ell-1}$, or $i+2^{\ell+1}$. Therefore, to calculate the total variation distance, it suffices to count the number of the problematic boundary cases, that is roughly $\frac{1}{\ell} = \eps$ portion of the possibilities.\footnote{For brevity and transparency of presentation we omit precise estimates of the boundary cases.}

\paragraph{Inductive Construction.}The general construction proceeds by induction on $n$. For each $n\ge 2$ we construct a distribution $\distset_n$ of $S\subset[N], |S|=n$ with $N= O\left( \frac{1}{\eps} \right) \uparrow \uparrow(n-1)$.
The base of inductive construction is specified above for $n=2$. For the inductive step, we assume that there is a distribution of $T=\{t_1<t_2<\cdots<t_{n-1}\}\sim\distset_{n-1}(\eps)$ with desired properties, 
where the maximum possible value of $t_{i}$ is $O\left(\frac{1}{\eps}\right)\uparrow\uparrow (n-2)$. 
We construct $S=\{s_1,s_2,\cdots,s_n\} \sim\distset_{n}$ as follows:
\begin{tcolorbox}[frame empty]
\begin{enumerate}
	\item Consider equivalent representation of $S$ as $(d_1,d_2,\ldots,d_n)$, where $d_i=s_{i}-s_{i-1}$ for $i\in[n]$ and $s_0=0$.
	\item Let $d_1\sim\uni \left[O(\frac{1}{\eps})\uparrow\uparrow (n-1) \right]$, and $(d_i\sim\uni[C^{t_{i-1}}])_{i=2}^n$ for $C=\frac{3}{\eps}$ independently from each other, where  $(t_i)_{i=1}^{n-1}$ are defined by $\{t_1,t_2,\cdots,t_{n-1}\} =T\sim\distset_{n-1}\left( \frac{\eps}{3} \right)$.
\end{enumerate}
\end{tcolorbox}

We first calculate the largest number used in the distribution $\distset_n$: 
\begin{multline*}
N = \max s_n =  \max \left( \sum_{i=1}^{n}\max d_i \right)\le O\left(\frac{1}{\eps} \right)\uparrow\uparrow (n-1) + \sum_{i=1}^{n-1} \left( \frac{3}{\eps} \right)^{\max t_i} \\
\le O\left(\frac{1}{\eps} \right)\uparrow\uparrow (n-1) + O\left(\frac{1}{\eps}\right)^{\max t_{n-1} +1} \le O\left(\frac{1}{\eps} \right)\uparrow\uparrow (n-1),
\end{multline*}
where the last inequality uses the induction hypothesis that the largest possible value of $t_{n-1}$ is $O\left(\frac{1}{\eps} \right)\uparrow\uparrow (n-2)$.

Next, we verify the stated total variation bound of the lemma. Consider $\setmi$ in the alternative representation for each $i\in[n]$: 
\begin{align*}
&\setmi=(d_1,\ldots,d_{i-1},d_{i}+d_{i+1},d_{i+2},\ldots,d_{n}), & \text{ for } i \le n-1\\
&\setmi[n]=(d_1,d_2,d_3,\ldots,d_{n-1}).
\end{align*}
We define auxiliary random sets $\Umi$ in the same alternative representation as $\setmi$ and independent distributions of all $d_i$'s:
\begin{align*}
& \Umi[1] \eqdef (d_1,d_3,\ldots,d_{n}), & \\
& \Umi \eqdef (d_1,\ldots,d_{i-1},d_{i+1},\ldots,d_{n}), & \text{ for } 2 \le i \le n
\end{align*}
Notice that
\begin{align*}
& \dtv\left(\setmi[1],\Umi[1]\right) = \dtv(d_1+d_2, d_1) \le \frac{C^{t_1}}{O(\frac{1}{\eps})\uparrow\uparrow (n-1)} < \frac{\eps}{3}, \\
& \dtv\left(\setmi,\Umi\right)= \dtv(d_{i}+d_{i+1}, d_{i+1}) \le \frac{C^{t_{i-1}}}{C^{t_i}} \le \frac{1}{C} < \frac{\eps}{3}, & \text{for } 2 \le i \le n-1 \\
& \dtv\left(\setmi[n], \Umi[n]\right) = 0.
\end{align*}
Here, the two inequalities hold by Lemma~\ref{lem:dtv_uni} for arbitrary fixed $t_1 \le O\left(\frac{1}{\eps}\right) \uparrow\uparrow (n-2)$ and for arbitrary fixed $t_{i-1} \le t_i -1$.
Next, we apply Lemma~\ref{lem:dtv_map} to the random mapping from $\Umi,\Umi[j]$ to $\Tsetmi[(i-1)], \Tsetmi[(j-1)]$ (or $\Tsetmi[(i-1)], \Tsetmi[1]$ when $j=1$) and get
\begin{align*}
& \dtv\left(\Umi,\Umi[j]\right) \le \dtv\left(\Tsetmi[(i-1)],\Tsetmi[(j-1)]\right) \le \frac{\eps}{3}, & \forall i,j \ge 2 \\
& \dtv\left(\Umi,\Umi[1]\right) \le \dtv\left(\Tsetmi[(i-1)],\Tsetmi[1]\right) \le \frac{\eps}{3}, & \forall i \ge 2
\end{align*}

Finally, we are ready to conclude the proof of the lemma. We consider two cases. First, we assume that $i,j\ge 2$ in the lemma's statement. Then
\begin{equation*}
\dtv\left(\setmi,\setmi[j]\right)\le  \dtv\left(\setmi,\Umi\right) + \dtv\left(\setmi[j],\Umi[j]\right) + \dtv\left(\Umi,\Umi[j] \right) \le \eps.
\end{equation*}
Second, we assume that $j=1, i\ge 2$. Then,  similar to the previous case we have
\begin{equation*}
\dtv\left(\setmi[1],\setmi\right)\le  \dtv\left(\setmi[1],\Umi[1]\right) + \dtv\left(\setmi,\Umi \right) + \dtv\left(\Umi[1],\Umi \right) \le \eps.
\end{equation*}
\end{proof}

Next, we prove that the same distribution from the above lemma with an amplified $N$ leads to the stronger property as stated in Lemma~\ref{lem:core_construction}.

\begin{proofof}{Lemma~\ref{lem:core_construction}}
We use the distribution $\distset(\frac{\eps}{n^2})$ constructed in Lemma~\ref{lem:main_construction}, which uses $N=O\left(\frac{n^2}{\eps}\right)\uparrow\uparrow (n-1)$. 
For a given pair of index sets $I$ and $J$, we iteratively construct a sequence of index sets $\{I_s\}, \{J_t\}$ in the following way:
\begin{tcolorbox}[frame empty]
\begin{itemize}
\item Let $I_0=I$ and $J_0 = J$ and $s=t=0$.
\item We continue the construction of the sequence until $I_s = J_t$. For each intermediate step, we write the elements in $I_s,J_t$ in ascending order:
\[
I_s = \{i_1, i_2, \ldots, i_k\}, \quad J_t = \{j_1, j_2, \ldots, j_k\}.
\]
\begin{itemize}
\item Let $i_r \ne j_r$ be the first different element. We have $i_\ell=j_\ell$ for $\ell \in [r-1]$.
\item If $i_r > j_r$, let $I_{s+1} = \{i_1,i_2, \ldots, i_{r-1}, i_r-1,i_{r+1}, \ldots, i_k\}$ and increase $s$ to $s+1$.
\item Else, let $J_{t+1} = \{j_1,j_2, \ldots, j_{r-1}, j_r-1,j_{r+1}, \ldots, j_k\}$ and increase $t$ to $t+1$.
\end{itemize}
\end{itemize}
\end{tcolorbox}
It is easy to see that the earth mover's distance between $I_s=\{i_1, \ldots,i_k\}$ and $J_t=\{j_1,\ldots,j_k\}$, i.e., the value of $\sum_{\ell \in [k]} |i_\ell - j_\ell|$ decreases by $1$ after each iteration, the above procedure ends after at most $n^2$ steps, since $\sum_{\ell \in [k]}|i_\ell - j_\ell| \le kn \le n^2$. Let there be $m_1$ different sets in $\{I_s\}$ and $m_2$ sets in $\{J_t\}$. We have $m_1+m_2 \le n^2$.

Each pair of $I_{s}$ and $I_{s+1}$ differs only by a single element: $i_r\in I_{s}, i_r-1\notin I_{s}$ and $i_r\notin I_{s+1}, i_{r}-1\in I_{s+1}$. Thus, we can express both $S_{I_s}, S_{I_{s+1}}$ as \emph{the same} (projection) function applied to $\setmi[(i_r-1)]$, or $\setmi[i_r]$, which deletes a subset of coordinates in either $\setmi[(i_r-1)]$, or $\setmi[i_r]$ with ranks $[n]\setminus\{i_1,\ldots,i_{r-1},i_r-1,i_{r},i_{r+1},\ldots,i_k\}$. By Lemma~\ref{lem:dtv_map}, $\dtv(S_{I_s},S_{I_{s+1}}) \le \dtv(\setmi[(i_r-1)],\setmi[i_r]) \le \frac{\eps}{n^2}$, due to the property from Lemma~\ref{lem:main_construction}.
Similarly, we also have $\dtv(S_{J_t},S_{J_{t+1}}) \le \frac{\eps}{n^2}$. Therefore, by triangle inequality for TV-distance
 \[
\dtv\left(S_I,S_J\right)\le \sum_{s=0}^{m_1-1}\dtv\left(S_{I_{s}},S_{I_{s+1}}\right) + \dtv\left(S_{I_{m_1}}, S_{J_{m_2}} \right) + \sum_{t=0}^{m_2-1}\dtv\left(S_{J_{t}},S_{J_{t+1}}\right) \le \frac{\eps}{n^2}\cdot n^2 \le \eps.
\]  
\end{proofof}

\section{Universal Construction}
\label{sec:universal}

In this section we show that cardinal online algorithms do not have advantage over ordinal algorithms for any ordinal task. We give a universal upper bound on the cardinal complexity 
building upon the construction from the previous section.

\begin{theorem}
\label{thm:main_online_alg}
Let $R:\Acts\times \sym(n) \times \sym(n)\to \R_{+}$ be the reward function of any $n$-round online ordinal task with distribution $\distpi$ of arrival orders $\pi$ and a distribution $\distsigma$ of element ranks $\sigma$. Then for any $\eps>0$, 
there exist $N=O\left(\frac{n^3 \cdot n!\cdot n!}{\eps}\right)\uparrow\uparrow(n-1)$ and a distribution $\distset$  over subsets of $[N]$, such that the advantage of the cardinal over ordinal algorithms is at most $1+\eps$, i.e.,
\begin{equation}
\max\limits_{\alg\in\Card} \Exlong[\substack{\pi \sim \distpi \\ \numbs=(S,\sigma) \sim\distsigma \times\distset}]{R(\alg( \numbs, \pi), \sigma, \pi)}\le
(1+\eps)\cdot \max\limits_{\alg\in\Ord}~~\Exlong[\substack{\pi\sim\distpi \\ \sigma \sim\distsigma}]{R(\alg(\sigma,\pi),\sigma,\pi)}
\label{eq:cardinal_complexity}
\end{equation}
\end{theorem}

\begin{proof}
Let $\distset\left(N\right)$ be the distribution from Lemma~\ref{lem:core_construction} with $N=O\left(\frac{n^3 \cdot n!\cdot n!}{\eps}\right)\uparrow\uparrow(n-1)$. Let us fix the cardinal algorithm $\alg^*$ with the best performance over values $\numbs=(S,\sigma) \sim\distsigma \times\distset$ and arrival orders $\pi \sim \distpi$. We shall construct an ordinal algorithm $\simul$ that \emph{simulates} behavior of $\alg^*$ on $\distset$ and achieves nearly the same expected reward on $\sigma\sim\distsigma$ and 
$\pi \sim \distpi$. The ordinal algorithm $\simul$ sees the identities $\pi[k]$ and the ranking  $\sigma^k$ of the first $k$ elements at each step $k$. We would like to simulate the result of $\alg^*(\numbts_{\pi[k]})$, where 
$\numbts=(S,\sigma)\sim\distset\times\distsigma$ and $\sigma^k=\sigma(\numbts_{\pi[k]})$. We need to be consistent across all $n$ steps. Hence, $\simul$ needs to use previously generated $\numbts_{\pi[k]}$ at step $k+1$.
\begin{tcolorbox}[frame empty]
\begin{itemize}
\item At step $1$, sample $S=\{s_1<\ldots< s_n\}\sim\distset$, let $\vt_{\pi(1)}=s_1$. Let $a_{1}=\alg_1^*(\vt_{\pi(1)})$. 
\item At step $k-1$, the ordinal algorithm $\simul$ took the same actions as $\alg^*(\pi[k-1],\numbts_{\pi[k-1]})$.
\item At step $k$, a new element $\pi(k)$ arrives and the ordinal algorithm $\simul$ sees the updated ranking $\sigma^{k}$ consistent with previous numbers $\numbts_{\pi[k-1]}$. Let $J\eqdef \{\sigma^{k}(j)\mid j< k\}$.
\item Sample $\tilde{S}=\{\tilde{s}_1<\ldots<\tilde{s}_n\}\sim\left(\distset ~\middle \vert~ \tilde{S}_{J}=\{\numbts_{\pi[k-1]}\},\sigma(\numbts_{\pi[k]})=\sigma^k \right)$ and set $\vt_{\pi(k)}=\tilde{s}_{\sigma^{k}(k)}$, so that $\numbts_{\pi[k]}=(\tilde{S}_{[k]},\sigma^{k})$.
Take the action $a_{k}=\alg_{k}^*(\pi[k],\numbts_{\pi[k]})$.
\end{itemize}
\end{tcolorbox}
The above construction of $\simul$ may sometimes fail at sampling $\tilde{S}\sim\left(\distset ~\middle \vert~ \tilde{S}_{J}=\{\numbts_{\pi[k-1]}\},\sigma^k \right)$, but as the next Lemma~\ref{lem:universal_lemma} shows, the probability of failure is negligibly small and the 
distribution of $\numbts_{\pi[k]}$ is close to $\numbs_{\pi[k]}$ at each step $k\le n$, where $\numbs=(S,\sigma)$ and $S\sim\distset$.
\begin{lemma}
\label{lem:universal_lemma}
For any $k\in[n]$, $\sigma,\pi\in\sym(n)$ the distance $\dtv\left(\{\numbts_{\pi[k]}\}, (S_{[k]} \mid S\sim\distset) \right)\le \frac{(k-1)\eps}{n\cdot n!\cdot n!}$.
\end{lemma}
\begin{proof}
We proceed by induction on $k$. For $k=1$ we choose the smallest number in the sampled set $\tilde{S}$ to be $\vt_{\pi(1)}=\tilde{S}_{[1]}=\{\tilde{s}_1\}$. Thus the distributions of  $\vt_{\pi(1)}$ and $S_{[1]}$ are exactly the same.
To verify the induction step for $k\ge 2$ we assume that the statement holds for $k-1$ and want to check it for $k$. At step $k$ we have $\tilde{S}_J=\{\vt_{\pi(1)},\ldots,\vt_{\pi(k-1)}\}$ and $\dtv\left(\tilde{S}_J,S_{[k-1]}\right)\le \frac{(k-2)\eps}{n\cdot n!\cdot n!}$ by the induction hypothesis.
Consider the true cardinal instance $\numbs=(S,\sigma)$ with the distribution $S\sim\distset$. By  Lemma~\ref{lem:core_construction} for $\eps'=\frac{\eps}{n\cdot n!\cdot n!}$ we have $\dtv\left(S_{[k-1]},S_J\right)\le\frac{\eps}{n\cdot n!\cdot n!}$. Hence, $\dtv\left(\tilde{S}_J,S_J\right)\le\frac{(k-1)\eps}{n\cdot n!\cdot n!}$ by triangle inequality for the TV-distance. As $\simul$ constructs $\vt_{\pi(k)}$ with the same distribution as $v_{\pi(k)}$ given $J\subset[k]$ and $S_{J}$, we have $\dtv\left(\tilde{S}_{[k]},S_{[k]}\right)=\dtv\left(\tilde{S}_J,S_J\right)\le\frac{(k-1)\eps}{n\cdot n!\cdot n!}$.
\end{proof}

We now conclude the proof of Theorem~\ref{thm:main_online_alg}. By Lemma~\ref{lem:universal_lemma} the simulation $\simul$ produces very similar results to $\alg^*$, i.e., $(\tacts(\sigma,\pi),\sigma,\pi)$ -- the actions  
of $\simul$ on any arrival order $\pi$ and ranking $\sigma$ are close in the TV-distance to the respective $(\acts(S,\sigma,\pi),\sigma,\pi)$ of the optimal cardinal algorithm $\alg^*$.
Therefore, we can compare the expected rewards of $\simul$ and $\alg^*$ as follows
\begin{gather}
\label{eq:universal_simulation}
\Ex[\numbs,\pi]{R(\acts(\numbs,\pi),\sigma(\numbs),\pi)}-\Ex[\sigma,\pi]{R(\tacts(\sigma,\pi),\sigma,\pi)}\le \dtv(\tacts(\numbts,\sigma,\pi),\acts(S,\sigma,\pi))\\
\cdot\Ex[\sigma,\pi]{\max_{\mathbf{b}\in\Acts}\left(R(\mathbf{b},\sigma,\pi)-0\right)}\le \dtv(\numbts,\numbs)\cdot\Ex[\sigma,\pi]{\max_{\mathbf{b}\in\Acts}R(\mathbf{b},\sigma,\pi)}
\le\frac{(n-1)\eps}{n\cdot n!\cdot n!}\cdot\Ex[\sigma,\pi]{\max_{\mathbf{b}\in\Acts}R(\mathbf{b},\sigma,\pi)}.\nonumber
\end{gather}
Next, we have a trivial ordinal algorithm that guesses the arrival order $\pi$ and ranking $\sigma$:
\[
(\pi^*,\sigma^*) = \argmax_{\pi,\sigma} \left( \prob[\distpi]{\pi} \cdot\prob[\distsigma]{\sigma} \cdot \max_{\acts\in\Acts} R(\acts,\sigma,\pi) \right),
\]
and then chooses corresponding optimal actions at each step. This algorithm achieves at least $\frac{1}{n!\cdot n!}$ fraction of the offline optimum $\expect[\sigma,\pi]{\max_{\acts}R(\acts,\sigma,\pi)}$, since there are at most $n!\cdot n!$ possible orders and rankings 
and we choose one with the maximal expected contribution. This means that
\be
\label{eq:trivial_alg}
\max\limits_{\alg\in\Ord}\Ex[\sigma,\pi]{R(\alg(\sigma,\pi),\sigma,\pi)} \ge\frac{1}{n!\cdot n!}\expect[\sigma,\pi]{\max_{\acts\in\Acts}R(\acts,\sigma,\pi)}.
\ee
We combine \eqref{eq:universal_simulation}, \eqref{eq:trivial_alg}, and the fact that $\simul$ is an ordinal algorithm to get the required inequality
\[
(1+\eps)\max\limits_{\alg\in\Ord}\Ex[\sigma,\pi]{R(\alg(\sigma,\pi),\sigma,\pi)} \ge \Ex[\numbs,\pi]{R(\acts(\numbs,\pi),\sigma(\numbs),\pi)}.
\]
\end{proof}

Theorem~\ref{thm:main_online_alg} states that for a sufficiently large size of the universe $[N]$, the optimal cardinal algorithm does not have much advantage over the best ordinal algorithm. However, the tower-of-exponents dependency of $N=\Omega\left(\frac{n}{\eps}\right)\uparrow\uparrow(n-1)$ on $n$ and $\eps$ is too impractical for any imaginable scenario. Hence, it is natural to ask for each given ordinal task (like selecting the maximum in the secretary problem) what is the minimal size of the universe $N$ such that the advantage of the cardinal over ordinal algorithms is at most $\eps$. We call this number $N$ the \emph{cardinal complexity} of a given ordinal task. 
\begin{definition}[Cardinal Complexity]
For a given online ordinal task and a parameter $\eps > 0$, its cardinal complexity is the minimum $N$ such that there is a distribution $\distset$ satisfying \eqref{eq:cardinal_complexity}.
\end{definition}

\section{Cardinal Complexity: General Lower Bound}
\label{sec:guessing}

We now establish a lower bound on the cardinal complexity of a natural die guessing game defined below, by designing an algorithm that efficiently utilizes the cardinal information in the game. 

\paragraph{Die Guessing.} The universe $\universe$ corresponds to the $n$ faces of a die and $\numbs$ corresponds to the numbers written on the die, that are $n$ distinct integers between $1$ and $N$. The distribution $\distpi$ is a uniform random order. The $n$ faces/numbers are shuffled according to $\pi \sim \distpi$. The action space $A_k$ is empty for $k \le n-2$ and $k=n$. At step $n-1$, the algorithm observes the first $n-1$ numbers and makes a guess from $A_{n-1} = \{1,2,\ldots,n\}$. The ordinal reward function $R$ equals $1$ when the algorithm guesses correctly $\pi(n)$ at step $n-1$ and equals $0$ otherwise. 
Since the $n$ faces of the die are symmetric, the adversary can choose a set $S \subseteq [N]$ and apply a uniform random permutation $\sigma$ for assigning the numbers to the faces. Against such instances, without loss of generality, we assume that the algorithm only depends on the set of the first $n-1$ numbers, and does not make use of the identities of the elements.

This game is an analogue of Lemma~\ref{lem:main_construction} in the universal construction. For technical reasons, we consider a slightly more general version of the die guessing game, which we call \emph{perturbed rank guessing}. 
The new game gives more power to the adversary, which is needed in our inductive proof later.

\paragraph{Perturbed Rank Guessing.} 
Given $n, N$, and a probability distribution $\vect{p}=(p_1,\ldots,p_n) \in \Delta_n$, the adversary (first player) chooses a set $S \subseteq [N]$ of $n$ distinct integers $s_1 < s_2 < \cdots < s_n$, with a technical condition that $s_{i} - s_{i-1} \ge 20$ for all $i \ge 2$. 
Then $\setmi$ is generated by deleting a single random number from $S$, where each $s_i$ is deleted with probability $p_i$.
Upon seeing $\setmi$, the adversary can modify every number of $\setmi$ by $\pm 1$ or $0$ and show modified numbers $\oset$ to the algorithm (second player). 
Finally, the algorithm guesses the index $i \in [n]$ of the deleted number $s_i\in S$. If the algorithm guesses correctly, the reward is $\frac{1}{p_i}$. Otherwise, the reward is $0$.

\begin{remark}
The perturbed ranking guessing game does not belong to the family of online ordinal tasks, as the adversary has an extra power to perturb each number before it is observed by the algorithm. On the other hand, this game is harder for the algorithm than the die guessing game, in the sense that if we have a cardinal algorithm $\alg$ for the perturbed rank guessing game, we can apply it to the die guessing game and achieve the same expected reward. Indeed, we first set all probabilities $p_i = \frac{1}{n}$ for every $i$; and
in order to meet the technical condition that gaps between consecutive numbers are at least $20$, we multiply each observed number by $20$ before we call $\alg$ as a black box. Effectively, we translate the instance from set $S=(s_1,\ldots,s_n)$ to $S' = (20 s_1, \ldots, 20 s_n)$. 
\\
\end{remark}

Our main result is a randomized algorithm with the following performance guarantee for the perturbed rank guessing game. We remark that we do not try to optimize the dependency on $n$. The most important regime for us is when $n$ is a constant and $N \to \infty$.

\begin{theorem}
\label{thm:guessing}
There exists an algorithm for the perturbed rank guessing game with expected reward (equals to the advantage of cardinal algorithm) $1 + \frac{1}{(6n)^{7n}} \cdot \Omega\left(\frac{1}{\klog[n-2] N} \right)$, where $\klog[n](x) \eqdef \underbrace{\log \log \dots \log x}_{n\textup{ logs}}$.
\end{theorem}
As a corollary, by applying the algorithm to the die guessing game as explained above, we establish a lower bound on the cardinal complexity.
\begin{corollary}
The cardinal complexity of the die guessing game is at least $\underbrace{2^{2^{\iddots^{\Omega\left(\frac{1}{(6n)^{7n}\eps} \right)}}}}_{n-2\text{ twos}}$.
\end{corollary}

The next corollary shows that the dependency on $\eps$ of the $n$-face die guessing game is a tower of exponents for any $n$ ($n$ is not necessarily a constant) of arbitrary constant height $c\le n-2$.
\begin{corollary}
For any constant $c\le n$ the cardinal complexity of the $n$-faces die guessing game is at least $\underbrace{2^{2^{\iddots^{\Omega\left(\frac{1}{\eps} \right)}}}}_{c-2\text{ twos}}$.
\end{corollary}
\begin{proof}
We reduce the $n$-face die guessing game to the $c$-face perturbed rank guessing game (the reduction does not use any perturbations, only the non-uniform probabilities $\vect{p}=(p_1,\ldots,p_c) \in \Delta_c$ for the hidden face). The algorithm for the $n$-face die guessing game works as follows: consider the largest $c-1$ numbers among $n-1$ visible faces and try to guess the rank of the hidden number relative to them using the algorithm for $c$-face perturbed rank guessing game with probabilities $\vect{p}=(\frac{n-c+1}{n},\frac{1}{n},\ldots,\frac{1}{n})$; if our guess in the $c$-face game is that the hidden number is the smallest number, then we pick our answer uniformly at random among the smallest $n-c+1$ numbers in the $n$-face game; otherwise we simply report the same rank as in the $c$-face game. This algorithm guesses correctly with probability $\frac{1}{n} \left( 1+ \Omega\left(\frac{1}{\klog[c-2] N} \right) \right)$ and concludes the proof of the corollary. We omit a straightforward calculation of the performance guarantee.
\end{proof}

Before we delve into technical details of
Theorem~\ref{thm:guessing} proof, we give a high level overview of our approach in the next subsection. This is the most technically involved part of our paper and the complete proof is provided in Appendix~\ref{app:guessing}.

\subsection{Proof Sketch}
\label{sec:alg_sketch}
Consider the alternative representation of set $S = (s_1, d_1, \ldots, d_{n-1})$, where $d_i = s_{i+1}-s_i$. We remark that we use different notations from the previous section as our algorithm shall not depend on the value of $s_1$.
After the deletion of a number, our algorithm observes $\oset = \{\os_1, \os_2, \ldots, \os_{n-1}\}$. 
Observe that the $n-1$ numbers partition $[N]$ into $n$ intervals $I_1 = [1,\os_1), I_2 = (\os_1, \os_2), \ldots, I_n = (\os_{n-1},N]$. It is equivalent between guessing the index $j$ of the deleted number and guessing which interval $I_j$ the deleted number belongs to. We shall describe our algorithm below as guessing the interval, which is more intuitive.
And for now, say we are playing the original die guessing game.

Our first step is to show that a hard instance must be like $d_1 \ll d_2 \ll \ldots \ll d_{n-1}$ (or $d_1 \gg d_2 \gg \ldots \gg d_{n-1}$).
We introduce two subroutines, $\monogaps$ (refer to Lemma~\ref{lem:monogaps}) and $\expgaps$ (refer to Lemma~\ref{lem:expgaps}) that achieve a constant advantage (that only depends on $n$ but does not depend on $N$) over ordinal algorithms, unless the instance has this specific shape, and perform not worse than any ordinal algorithm for this case.

Our second step focuses on instances with $d_1 \ll d_2 \ll \ldots \ll d_{n-1}$. We use $g_i = \os_{i+1} - \os_i, i \in [n-2]$ to denote the gaps observed by our algorithm. 
Our recursive algorithm only looks at those gaps and views them as a random (perturbed) subset of $\{d_i\}_{i \in [n-1]}$ with $n-2$ numbers. E.g., when $s_i$ is deleted, the gaps we observe are
\[
(d_1, \ldots, d_{i-2}, d_{i-1}+d_i, d_{i+1}, \ldots, d_{n-1}) \approx (d_1, \ldots, d_{i-2}, d_i, d_{i+1}, \ldots, d_{n-1}), \quad \text{since } d_{i-1} \ll d_i.
\]
 We formalize this idea by taking the logarithm of $g_i$'s. Then we can treat $\{ \lfloor \log_2 g_i \rfloor \}_{i \in [n-2]}$ as a random subset of $\{ \lfloor \log_2 d_i \rfloor \}_{i \in [n-1]}$, within a tiny error of at most $1$:
\begin{align*}
& (\lfloor \log_2 d_1 \rfloor, \ldots, \lfloor \log_2 d_{i-2} \rfloor, \lfloor \log_2 (d_{i-1}+d_i) \rfloor, \lfloor \log_2 d_{i+1} \rfloor, \ldots, \lfloor \log_2 d_{n-1} \rfloor) \\
& = (\lfloor \log_2 d_1 \rfloor, \ldots, \lfloor \log_2 d_{i-2} \rfloor, \lfloor \log_2 d_i \rfloor + 0/1, \lfloor \log_2 d_{i+1} \rfloor, \ldots, \lfloor \log_2 d_{n-1}) \rfloor
\end{align*}
That is the reason why we introduced perturbation to the setting.
Moreover, notice that the $n$ possible deletions of $\{s_i\}_{i \in [n]}$ result in only $n-1$ possible gap vectors. Indeed, the two cases when $s_1$ or $s_2$ is deleted lead to (almost) the same set of observed gaps:
\[
(\lfloor \log_2 (d_1 + d_2) \rfloor, \lfloor \log_2 d_3 \rfloor, \ldots, \lfloor \log_2 d_{n-1} \rfloor) = (\lfloor \log_2 d_2 \rfloor + 0/1, \lfloor \log_2 d_3 \rfloor, \ldots, \lfloor \log_2 d_{n-1} \rfloor)
\]
In particular, a uniform deletion of the $n$ numbers from $S$ leads to a non-uniform deletion of the $n-1$ gaps with probabilities $\{\frac{2}{n}, \frac{1}{n}, \ldots, \frac{1}{n}\}$.
This is why we consider non-uniform deletion of the numbers in the perturbed version of the die guessing game.
Those changes to the setting do not affect too much the analysis for the first step, but allow us to strengthen our induction hypothesis in the second step.
Finally, notice that our recursive step reduces the cardinal complexity $N$ by applying a logarithmic function after each step, from which we derive the stated algorithm's performance with iterative logarithms. 
\section{Cardinal Complexity: Game of Googol}
\label{sec:max}
The universal construction from Section~\ref{sec:universal} works for arbitrary online ordinal tasks such as the Game of Googol, albeit it is prohibitively large. The Rank Guessing game from Section~\ref{sec:guessing} is a core problem showing that the later inefficiency is unavoidable in general. In this section we show that the specific online problem of the Game of Googol admits much smaller construction with $N=O\left(\left(\frac{n}{\eps}\right)^n\right)$. 
 
As before, it will be useful to analyze first a respective single-shot game -- a natural modification of the Rank Guessing game -- the \emph{Maximum Guessing game}.

\paragraph{Maximum Guessing Game.} 
Given $n, N$, the adversary (first player) chooses a set $S \subset [N]$ of $n$ distinct integers $s_1 < s_2 < \cdots < s_n$. 
Then $\setmi$ is generated by deleting uniformly at random a single number from $S$. The algorithm (second player) sees $\setmi$ and guesses whether the deleted number is the maximum in $S$ or not, i.e., guesses whether $i=n$ (``yes''), or $i\ne n$ (``no'').
If the ``yes'' guess $i=n$ is correct, the algorithm's reward is $n$; if the ``no'' guess $i\ne n$ is correct, the reward is $\frac{n}{n-1}$. Otherwise, if the algorithm's ``yes'' or ``no'' guess is incorrect, then the reward is $0$.
This game shares the obvious common feature with the Game of Googol, that we only need to guess whether a number is the largest or not.
The expected reward of any ordinal algorithm for the maximum guessing game is $1$. Indeed, the ordinal algorithm does not see the numbers of $\setmi$, i.e., its decision does not depend on the input. The guess $i=n$ wins with probability $\frac{1}{n}$ with the expected reward of $1=n\cdot\frac{1}{n}$; the guess $i\ne n$ wins with probability $\frac{n-1}{n}$ with the same expected reward of $1=\frac{n}{n-1}\cdot\frac{n-1}{n}$.

\paragraph{Level Setting.} Similar to the Rank Guessing game our goal is to construct a distribution $S\sim\distset$ such that the expected reward of a cardinal algorithm is not better than the reward of the ordinal algorithm plus a small $\eps$. Recall that in the construction for the Rank Guessing game it was useful to have gaps $\{d_i=s_i-s_{i-1}\}_{i=1}^{n}$ (where $s_0=0$) of different magnitudes, so that when we merge gaps $d_i$ and $d_{i+1}$ together by deleting $i$-th element the distribution of $d_{i}+d_{i+1}$ is virtually indistinguishable from the distribution $\max\{d_i,d_{i+1}\}$. We achieved this property by having different \emph{levels} of gaps with $d_i\sim \uni [N_i]$ where $N_2 \ll N_3 \ll \ldots \ll N_n \ll N_1$. We use a similar approach for Max Guessing game. Namely,
we  have $d_i\sim L_{j_i}$ with different $L_{j_i}\in \{L_i \eqdef \uni [\Delta^i]\}_{i=1}^n$ for a large $\Delta$.

\subsection{Construction of the distribution $\distset$}
\label{sec:construction_googol}
Our construction of the distribution $S\sim\distset$ relies on a set of distributions $\{L_i \eqdef \uni [N_i]\}$, where $N_i = \Delta^{i-1}$, and a distribution 
$\distlev$ over permutation of levels $\lev\in\sym(n)$. The $i$-th gap is generated as $d_i=s_i-s_{i-1}\sim L_{\levi}$ (we assume $s_0=0$). The first gap $d_1=s_1$ has the largest level $\levi[1]=n$. We obtain distribution of level sequences $\lev\sim\distlev$ as the stationary distribution of a Markov chain on $\sym(n)$ that we shall specify later. The construction of  $S\sim\distset$ is as follows.
\begin{tcolorbox}[frame empty]
\begin{enumerate}
\item Consider representation of $S$ as $(d_1,d_2,\ldots,d_{n})$, where $d_i=s_{i+1}-s_{i}$ ($s_0=0$).
\item Sample $\lev =(\levi[1],\levi[2],\ldots,\levi[n])\in \sym(n)$ from $\distlev$ ($\levi[1]=n$).
\item Sample $d_i \in [N_{\levi}]$ from $L_{\levi}$ independently.
\end{enumerate}	
\end{tcolorbox}

The deletion of $i$-th number $s_i$ and consequent merge of the gaps $d_{i-1}$ and $d_i$ corresponds to ``merging of levels'' $L_{\levi}$ and $L_{\levi[i+1]}$ into the level $L_{\max\{\levi,\levi[i+1]\}}$. I.e., we have the following algebraic operation applied to the set of permutation in the support of $\lev\sim\distlev$.
\begin{align*}
s_i \text{ is deleted:} \quad & \rhomi \eqdef (\levi[1],\ldots,\levi[i-1], \max (\levi, \levi[i+1]), \levi[i+2],\ldots, \levi[n]), \quad i\in [1,n-1] \\
s_n \text{ is deleted:} \quad & \rhomi[n] \eqdef (\levi[1],\levi[2],\ldots,\levi[n-1]) 
\end{align*}
In the Maximum Guessing game we would like to have the distribution $\rhomi[n]$ of $\lev\sim\distlev$ to be indistinguishable from the distribution of $\rhomi$ for $i\sim\uni [n-1]$ and $\lev\sim\distlev$. The latter property can be captured by a linear algebraic equation on the set of permutations $\lev\in\{\sym(n) | \levi[1]=n \}$. 
We consider the following Markov chain defined by $(n-1)! \times (n-1)!$ matrix $\vect{M}$ with indexes $\lev,\lev' \in \{\sym(n)|\levi[1]=\levi[1]'=n\}$:
\[
M(\lev, \lev') \eqdef \frac{ \lvert \{ i\in[n-1] \mid \rhomi = \rhomi[n]' \}\rvert }{n-1}.
\]
We view $\vect{M}$ as a transition matrix on the state space $\sym(n-1)=\{\lev\in\sym(n) | \levi[1]=n \}$ and let $\vect{p}$ be its stationary distribution, i.e., $\vect{p\cdot M} = \vect{p}$. We define $\distlev\eqdef\vect{p}$.
We give a concrete example below to illustrate the construction. 
\begin{tcolorbox}[frame empty]
\paragraph{Example.}
Consider the case when $n=4$, the transition matrix $\vect{M}$ is given below:
\[
\vect{M} = 
\begin{blockarray}{ccccccc}
& 412(3) & 413(2) & 421(3) & 423(1) & 431(2) & 432(1)\\
\begin{block}{c(cccccc)}
  4123 & 0 & 1/3 & 0 & 2/3 & 0 & 0 \\
  4132 & 0 & 1/3 & 0 & 0 & 0 & 2/3 \\
  4213 & 0 & 1/3 & 0 & 2/3 & 0 & 0 \\
  4231 & 0 & 0 & 0 & 1/3 & 2/3 & 0 \\
  4312 & 1/3 & 0 & 0 & 0 & 0 & 2/3 \\
  4321 & 0 & 0 & 1/3 & 0 & 1/3 & 1/3 \\
\end{block}
\end{blockarray}
 \]
E.g., consider the first row of the matrix corresponding to the gap levels $\lev=(4,1,2,3)$ of $(d_1,d_2,d_3,d_4)$. When the first or the second number $s_1,s_2$ is deleted, the level gaps observed by the algorithm would be $(4,2,3)$; when the third number $s_3$ is deleted, the level gaps observed by the algorithm would be $(4,1,3)$. We don't consider deletion of the last number $s_n$, which corresponds to permutation $\lev'=(4123)$ (we write instead $\levi[\text{-}4]'=412$) in the first column of $\vect{M}$.
The stationary distribution of the above transition matrix is 
\[
\vect{p} = \left( \frac{5}{66}, \frac{6}{66}, \frac{7}{66}, \frac{2}{11}, \frac{5}{22}, \frac{7}{22} \right)
\]
\end{tcolorbox}
Let $\levi[\text{-uni}]$ be a random permutation $\rhomi$ with $\lev \sim \vect{p}$ and $i \sim \uni[n-1]$. 
It turns out that $\levi[\text{-uni}]$ is indistinguishable from $\rhomi[n]$, for $\lev \sim \vect{p}$.
\begin{claim}
\label{cl:permutations_stationary}
The distribution $\levi[\text{-uni}]$ is the same as the distribution $\rhomi[n]$ for $\lev\sim\distlev$.
\end{claim}
\begin{proof}
Fix an arbitrary $\lev' \in \{\sym(n) |\levi[1]'=n\}$, we have
\begin{multline*}
\prob[\lev\sim\distlev]{\levi[\text{-uni}] = \levi[\text{-}n]'} = \prob[\lev,i]{\rhomi = \levi[\text{-}n]'} = \sum_{\lev} p_{\lev} \cdot \sum_{i \in [n-1]} \frac{\ind{\rhomi[i] = \levi[\text{-}n]}}{n-1} \\ 
= \sum_{\lev} p_{\lev} \cdot M(\lev, \lev') = p_{\lev'} =
\prob[\lev\sim\distlev]{\lev=\lev'}=
\prob[\lev\sim\distlev]{\rhomi[n]=\levi[\text{-}n]'},
\end{multline*}
where the third and forth equalities follow from the definition of $\vect{M}$ and $\vect{p}$.
\end{proof}
The Claim~\ref{cl:permutations_stationary} implies that one cannot tell apart whether we deleted $s_i$ with $i\sim\uni[n-1]$, or if we deleted the maximum $s_n$ by looking at the $n-1$ gap levels. Thus our construction of $S\sim\distset$ should work for Maximum Guessing game. It is rather straightforward to formally verify that the cardinal algorithm $\alcard$ has an advantage of at most $\eps$ over the ordinal algorithm\footnote{One just need to give upper bounds on the total variation distances between $d_i+d_{i+1}$ and $L_{\max\{\levi,\levi[i+1]\}}$.} $\alord$ for Max Guessing game. We omit this verification, as our actual focus is on the Game of Googol.
Moreover, it is also easy to show that the construction $\distset$ works at any step $k\le n$ of the Game of Googol, i.e., at any step $k$ the cardinal algorithm $\alcard$ upon observing $k$ numbers does not have any significant advantage over the ordinal algorithm $\alord$ in a single-shot game of guessing whether the current maximum is the global maximum. Interestingly, the latter property is not enough to guarantee that $\alcard$ is not significantly better than $\alord$ in the Game of Googol. Indeed, imagine that $\alcard$ sees the arrival of a new current maximum number $s^*$ at step $k\approx\frac{2}{5}n$. The ordinal algorithm  would take $s^*$ (as $\frac{2}{5}n>\frac{n}{e}$) and win with probability $\frac{2}{5}$. The $\alcard$ may take $s^*$ which for the construction $\distset$ also wins with probability $\frac{2}{5}$; it may also skip $s^*$ sometimes, e.g., when $\alcard$ knows that $s^*$ is among top two largest numbers, which leads to a winning probability of $\frac{3}{5}$ for $\alcard$ (when $s^*$ is not the global maximum, $\alcard$ will see the global maximum with remaining probability $\frac{3}{5}$). More generally, it is unclear why a cardinal algorithm $\alcard$ cannot make a better than $\alord$ guesses about the global rank of the current maximum value in the construction $\distset$, which may help it to get an advantage over $\alord$. Thus we have to analyze online algorithm $\allev$ with observed levels and prove that it has no advantage over the best ordinal online algorithm.

\subsection{Secretary with levels}
\label{sec:level_setting}
In this section, we focus on the level setting for online algorithms and prove that cardinal algorithms have no advantage over ordinal algorithms.
Since $n$ cards are symmetric, the adversary may choose  uniformly at random an assignment $\sigma$ of $n$ numbers to $n$ cards so that the card identities reveal no extra information. We interpret the arrival order $\pi$ differently from the previous sections 
to simplify the notations: $\pi(i)$ denotes the rank of the $i$-th arriving number among the $n$ numbers, instead of its identity.

\paragraph{Secretary with Levels.} 
Consider the following variant of game of googol:
\begin{enumerate}
\item A level vector $\lev$ is drawn from $\distlev$ and an arrival order $\pi$ is drawn uniformly from $\sym(n)$.
\item At each step $k \in [n]$, the new number $\pi_k$ arrives. Let $\pi[k] = \{i_1 < i_2 < \cdots < i_k\}$ be the set of arrived numbers.
\item The online algorithm observes relative ranks $\pi^k \in \sym(k)$ of $\pi[k]$ and the corresponding level vector 
\[
\lev^k = \left(\max_{j \le i_1} \levi[j], \max_{i_1<j\le i_2} \levi[j], \cdots, \max_{i_{k-1} < j \le i_k} \levi[j] \right)
\]
\item The algorithm decides whether to stop and accept the $k$-th number based on $\pi^k, \lev^k$. 
\item The goal is to maximize the probability of taking global maximum, i.e., $\pi_k=n$.
\end{enumerate}

The level setting is an idealized version of the Game of Googol. 
Instead of observing $\{s_{i_1}, s_{i_2}, \cdots, s_{i_k} \}$ at step $k$, the algorithm $\allev$ only observes the levels of the gaps $\rho^k$. 
Suppose $\lambda \in \Lambda^k$ is a partial $k$-permutation over $[n]$. Define the following ``deletion'' operators of the $i$-th element of $\lambda$:
\begin{align*}
& \di(\lambda) \eqdef (\lambda(1),\ldots,\lambda(i-1), \max (\lambda(i), \lambda(i+1)), \lambda(i+2),\ldots, \lambda(k)), \quad i\in [k-1] \\
& \di[k](\lambda) \eqdef (\lambda(1),\lambda(2),\ldots,\lambda(k-1)) 
\end{align*}
We use $\udi(\lambda)$ to denote a uniform deletion of one of the $k$ elements in $\lambda$, i.e., $\udi(\lambda)$ is drawn uniformly at random from $\{\di(\lambda)\}_{i \in [k]}$.
Similarly, we use $\undi(\lambda)$ to denote a uniform deletion of one of the $k-1$ elements in $\lambda$ except for the maximum one, i.e., $\undi(\lambda)$ is drawn uniformly at random from $\{\di(\lambda)\}_{i \in [k-1]}$.
The following Lemma~\ref{lem:googol_deletion} allows us to strengthen the guarantee from Claim~\ref{cl:permutations_stationary} to the partial $k$-permutation of levels at any step $k$.
\begin{lemma}
\label{lem:googol_deletion}
For any $k \in [n]$,
\[
\underbrace{\udi(\udi(\cdots \udi(\distlev)\cdots)}_{k \textup{ times}} = \di[n-k+1]( \underbrace{\undi(\undi(\cdots \undi(\distlev)\cdots)}_{k-1 \textup{ times}} ) =\underbrace{\undi(\undi(\cdots \undi(\distlev)\cdots)}_{k \textup{ times}}
\]
\end{lemma}
\begin{proof}
We prove the statement by induction. For the base case of $k=1$, the statement that $\udi(\distlev)= \di[n](\distlev)=\undi(\distlev)$ holds by Claim~\ref{cl:permutations_stationary} and an observation that $\udi$ is a mixture of operator $\di[n]$ with probability $\frac{1}{n}$ and operator $\undi$ with probability $\frac{n-1}{n}$.
Suppose the statement holds for $k$. Then we have
\[
	\udi\left(\underbrace{\udi(\cdots \udi(\distlev)\cdots)}_{k \textup{ times}} \right) = \udi \left( \di[n-k+1]( \underbrace{\undi(\cdots \undi(\distlev)\cdots)}_{k-1 \textup{ times}})~\right)
	= \di[n-k] \left(\underbrace{\undi(\cdots \undi(\distlev)\cdots)}_{k \textup{ times}} \right).
\]
Here, the first equality holds by induction hypothesis; the second equality holds as for any $\lambda \in \Lambda^{n-k}$, we have $\udi(\di[n-k+1](\lambda)) = \di[n-k](\undi(\lambda))$. We also have by induction hypothesis  
\[
  \di[n-k] \left(\underbrace{\undi(\cdots \undi(\distlev)\cdots)}_{k \textup{ times}} \right)=
	\udi\left(\underbrace{\udi(\udi(\cdots \udi(\distlev)\cdots)}_{k \textup{ times}} \right)=
	\udi\left(\underbrace{\undi(\undi(\cdots \undi(\distlev)\cdots)}_{k \textup{ times}} \right).
\]
Now, since the last operator $\udi(\lambda)$ is a mixture of $\di[n-k](\lambda)$ (with probability $\alpha=\frac{1}{n-k}$) and $\undi(\lambda)$
(with probability $1-\alpha=\frac{n-k-1}{n-k}$) for all $\lambda \in \Lambda^{n-k}$, we have that
\begin{multline*}
\di[n-k] \left(\underbrace{\undi(\cdots \undi(\distlev)\cdots)}_{k \textup{ times}} \right)=
\left(\underbrace{\alpha\di[n-k]+(1-\alpha)\undi}_{\udi}\right)\left(\underbrace{\undi(\undi(\cdots \undi(\distlev)\cdots)}_{k \textup{ times}} \right) \implies\\
\di[n-k] \left(\underbrace{\undi(\cdots \undi(\distlev)\cdots)}_{k \textup{ times}} \right) =
\undi\left(\underbrace{\undi(\undi(\cdots \undi(\distlev)\cdots)}_{k \textup{ times}} \right).
\end{multline*}
This concludes the proof of the lemma.
\end{proof}

Next, we prove our main theorem that the optimal online algorithm $\allev$ in the level setting accepts the maximum number with the same probability as the optimal ordinal algorithm $\alord$ for the secretary problem.
\begin{theorem}
\label{thm:level_ordinal}
$\Ex[\pi,\lev\sim\distlev]{\allev(\pi,\lev)} = \Ex[\pi]{\alord(\pi)}$.
\end{theorem}

\begin{proof}
We first state the standard backward induction analysis of the optimal algorithm of the ordinal secretary problem. 

Suppose the game has reached the $i$-th step and the algorithm has not accepted any number.  
Let $f_i(\pi^i)$ be the expected winning probability of the optimal algorithm when the relative ranks among the first $i$ numbers are $\pi^i$. 
Then, we have
\begin{align*}
& f_n(\pi^n) = \ind{\pi^n(n)=n} \\
& f_i(\pi^i) = \Exlong[\pi]{f_{i+1}(\pi^{i+1}) \mid \pi^i} & \text{if } \pi^i(i) \ne i, \forall i \in [n-1] \\
& f_i(\pi^i) = \max \left( \frac{i}{n}, \Exlong[\pi]{f_{i+1}(\pi^{i+1}) \mid \pi^i} \right) & \text{if } \pi^i(i) = i, \forall i \in [n-1]
\end{align*}
The first equation corresponds to the base case when the algorithm reaches the last step. The winning probability is either $0$ or $1$, depending on whether the $n$-th number is of the largest rank. The second equation corresponds to the case when the $i$-th number is not the maximum so far. Then, the optimal algorithm should reject it and continues to the next step. The third equation corresponds to the case when the $i$-th number is the maximum so far. The optimal algorithm either accepts it and wins with probability $\frac{i}{n}$, or rejects it and continues to the next step. The winning probability of this algorithm is $\Ex[\pi]{\alord} = \Ex[\pi]{f_1(\pi^1)}$.

Next, we switch to the level setting. Suppose the game has reached the $i$-th step and has not accepted any number. Let $g_i(\lev^i,\pi^i)$ be the optimal winning probability when the relative ranks are $\pi^i$ and the observed gap levels are $\lev^i$. 
Notice that the optimal algorithm can also make use of information observed in earlier steps, i.e. $\{\lev^j,\pi^j\}_{j < i}$. However, this information is induced by $\lev^i, \pi^i$, and hence, we omit these parameters in our formulation. Similar to the analysis of the optimal ordinal algorithm, we have the following recursive formulas.
\begin{align*}
& g_n(\lev^n,\pi^n) = \ind{\pi^n(n)=n} \\
& g_i(\lev^i,\pi^i) = \Exlong[\pi,\lev]{g_{i+1}(\lev^{i+1},\pi^{i+1}) \mid \lev^i,\pi^i} & \text{if } \pi^i(i) \ne i, \forall i \in [n-1] \\
& g_i(\lev^i,\pi^i) = \max \left( \Prlong[\pi,\lev]{\pi(i)=n \mid \lev^i,\pi^i}, \Exlong[\pi,\lev]{g_{i+1}(\lev^{i+1},\pi^{i+1}) \mid \lev^i, \pi^i} \right) & \text{if } \pi^i(i) = i, \forall i \in [n-1]
\end{align*}
The winning probability of this optimal algorithm is $\Ex[\pi,\distlev]{\allev(\pi,\lev)} = \Ex[\pi,\distlev]{g_1(\pi^1,\lev^1)}$.

Next, we shall prove that for an arbitrary realization of $\lev^i=\mu^i \in \Lambda^i, \pi^i=\sigma^i \in \sym(i)$, 
\begin{equation}
\label{eq:level}
g_i(\mu^i, \sigma^i) = f_i(\sigma^i)	
\end{equation}

As a consequence of this statement, we conclude the proof of the theorem by
\[
\Exlong[\pi,\lev]{\allev(\pi,\lev)} = \Exlong[\pi,\lev]{g_1(\pi^1,\lev^1)} = \Exlong[\pi]{f_1(\pi^1)} =\Exlong[\pi]{\alord}
\]

We prove \eqref{eq:level} by backward induction on $i$. The base case when $i=n$ holds according to the definition of the two quantities. By comparing the recursive formulas of $f_i$ and $g_i$, it suffices to verify that when $\sigma^i(i)=i$, we have that
\[
\Prlong[\pi,\lev]{\pi(i)=n \mid \lev^i = \mu^i, \pi^i = \sigma^i} = \frac{i}{n}
\]
That is, when the $i$-th number is maximum so far, then \emph{for any realized sequence of levels} $\mu^i$, by accepting $i$ we win with probability $\frac{i}{n}$. Indeed,
\begin{multline*}
\Prlong[\pi,\lev]{\pi(i)=n \mid \lev^i = \mu^i, \pi^i = \sigma^i} 
= \frac{\Prlong[\pi,\lev]{\pi(i)=n, \lev^i =\mu^i \mid \pi^i = \sigma^i}}{\Prlong[\pi,\lev]{\lev^i=\mu^i \mid \pi^i = \sigma^i}}\\
=\Prlong[\pi]{\pi(i)=n \mid \pi^i=\sigma^i} \cdot \frac{\Prlong[\pi,\lev]{\lev^i =\mu^i \mid \pi(i)=n, \pi^i = \sigma^i}}{\Prlong[\pi,\lev]{\lev^i=\mu^i \mid \pi^i = \sigma^i}} = \frac{i}{n} \cdot \frac{\Prx[\lev]{\overbrace{\undi(\undi(\cdots\undi(\lev)\cdots)}^{n-i\text{ times}}=\mu^i}}{\Prx[\lev]{\underbrace{\udi(\udi(\cdots\udi(\lev)\cdots)}_{n-i \text{ times}}=\mu^i}} = \frac{i}{n},
\end{multline*}
where the last inequality follows from Lemma~\ref{lem:googol_deletion}. This concludes the proof of the theorem.
\end{proof}

\subsection{Final Proof Step}
\label{sec:simulation}
Finally, we conclude an upper bound of the cardinal complexity of the game of googol by combining Theorem~\ref{thm:level_ordinal} and the idea of merging gaps with different levels.

\begin{theorem}
	\label{thm:googol}
	The cardinal complexity of the game of googol is at most $N=O\left( \frac{n}{\eps}\right)^n$.
\end{theorem}

The intuition behind our construction of $\distset$ (in Section~\ref{sec:construction_googol}) is that merged gaps $d_i$ and $d_{i+1}$ (when $i$-th element is deleted) is  almost indistinguishable from $\max(d_i,d_{i+1})$. Intuitively, this allows one to reason about online algorithms in the idealized level setting, where any cardinal algorithms has no advantage over the ordinal algorithm for the distribution $\lev\sim\distlev$. 

To make it formal, we apply a similar simulation argument as we did in the proof of Theorem~\ref{thm:main_online_alg}. Specifically, for an arbitrary cardinal algorithm for the game of googol, we show that we can simulate it in the level setting with only $\eps$ decrement in the winning probability, which then induces an upper bound on the cardinal complexity of the game of googol. The proof is given in Appendix~\ref{app:googol}.

\section{Open Questions}
\label{sec:open}
As we initiated the study of cardinal complexity of online problems, there are many interesting open questions for future research.

\begin{enumerate}
\item Our construction $\distset$ for the secretary problem of $O\left(\left(\frac{n}{\eps}\right)^n\right)$, while much better than the 
cardinal complexity of the Maximum Guessing game, is still a rather large number. On the positive side, our guarantee works for any given $n$ and $\eps$. While it might be difficult to improve the dependency on $\eps$ for a fixed $n$, we have not considered the cardinal complexity of the secretary problem asymptotic in both $\eps$ and $n$. Specifically, it is natural to ask if there is a construction with cardinal complexity $O(\text{poly}(n,\eps))$ such that no cardinal algorithm can achieve winning probability of $1/e+\eps$ (i.e., consider the regime when $n\to\infty$). 
\item Study the cardinal complexity of the $(J,K)$-secretary problem. More broadly, it is interesting to find 
examples of problems with intermediate cardinal complexity between the secretary problem $\exp(O(\frac{1}{\eps}))$ and Rank Guessing game $O(\frac{1}{\eps}) \uparrow\uparrow (n-1)$. 
\item 
Our work suggests that limited cardinal complexity, i.e., a small support size $N$ of input numeric values can help online algorithms (both cardinal and/or ordinal) to improve their performance. The previous work on various random arrival models very often ignores the cardinal complexity considerations, e.g., by assuming that the numeric values can be arbitrary real numbers. 
In the mean time, it seems that many practical scenarios should have a rather small parameter for the cardinal complexity. To illustrate this point, consider for example the reviewing process at a computer science conference, in which a PC member needs to evaluate a pile of roughly $50$ papers. It is practically impossible to come up with a complete ranking of these many papers. Instead we tend to use a much smaller number of categories (types) from ``strong accept'' to ``strong reject'' to describe the papers. As a concrete problem, one can study the secretary problem/game of googol with a fixed number $N$ of types (naturally we allow ties) and design online cardinal or ordinal algorithms.
\end{enumerate}

\bibliographystyle{plain}
\bibliography{online}

\newpage
\appendix

\section{Proof of Lemma~\ref{lem:dtv_uni}}
\label{app:dtv}
We calculate the total variation distance directly. 
For every $x \in [\alpha_1+\beta_2, \beta_1]$, we have that
\[
\prob{x_1+x_2 = x} = \sum_{i = \alpha_2}^{\beta_2} \prob{x_1 = x-i, x_2 = i} = \sum_{i=\alpha_2}^{\beta_2} \frac{1}{\beta_1-\alpha_1+1} \cdot \frac{1}{\beta_2-\alpha_2+1} = \frac{1}{\beta_1-\alpha_1+1}.
\]
Thus,
\begin{multline*}
\dtv(x_1, x_1+x_2) = \frac{1}{2} \cdot \sum_{x=\alpha_1}^{\beta_1+\beta_2} \lvert \prob{x_1 = x} - \prob{x_1+x_2 = x} \rvert \\
= \sum_{x=\alpha_1}^{\alpha_1+\beta_2-1} \lvert \prob{x_1 = x} - \prob{x_1+x_2 = x} \rvert \le \sum_{x=\alpha_1}^{\alpha_1+\beta_2-1}\frac{1}{\beta_1-\alpha_1+1} = \frac{\beta_2}{\beta_1-\alpha_1+1}.
\end{multline*} 
\section{Proof of Theorem~\ref{thm:guessing}}
\label{app:guessing}
We will use notation $\oset$ for the set of ordered values that cardinal algorithm observes before making its guess.

\paragraph{Ordinal Algorithm.} An ordinal algorithm cannot see the numbers of $\oset$, i.e. its decision does not depend on the input. On the other hand, notice that the gaps between consecutive numbers are at least $20$. The perturbation would not change the relative order of the numbers. Hence, the guess of $i$ wins with probability $p_i$ with the expected reward of $1=p_i \cdot \frac{1}{p_i}$, for all $i \in [n]$. Therefore, any ordinal algorithm has expected reward of $1$.

\paragraph{Cardinal Algorithms.}
We prove that an algorithm can achieve a noticeably better reward than $1$ using $\oset$. 
Within this section, we are interested in the regime when $N \to \infty$ and $n$ is a small constant.
As a warm up we first describe the algorithms for $n=2,3$.

\paragraph{Warm-up for $n=2$.} 
Consider the following algorithm: it sees $\oset = \{\os\}$ and guesses $1$ with probability $\frac{\os}{N}$ and $2$ otherwise. 
Let $S=(s_1,s_2)$ be the ordered set chosen by the adversary. With probability $p_1$, we see $\os \in \{s_2 \pm 1,s_2\}$ and win with probability $\frac{\os}{N} \ge \frac{s_2-1}{N}$ and get reward $\frac{1}{p_1}$. 
With probability $p_2$, we see $\os \in \{s_1\pm 1,s_1\}$ and win with probability at least $1-\frac{s_1+1}{N}$ and get reward $\frac{1}{p_2}$. Therefore, the expected reward $\alg$ of the algorithm is at least 
\[
\alg \ge p_1 \cdot \frac{s_2-1}{N} \cdot \frac{1}{p_1} + p_2 \cdot \left( 1-\frac{s_1+1}{N} \right) \cdot \frac{1}{p_2} = 1 + \frac{s_2-s_1-2}{N} \ge 1+\frac{18}{N}.
\]

\paragraph{Warm-up for $n=3$.} 
Consider the following algorithm: it sees $\oset = (\os_1, \os_2)$ and guesses $2$ with probability $\frac{\log_2 (\os_2 - \os_1)}{\log_2 N}$ and guesses $1,3$ uniformly at random otherwise. 
Let $S=\{s_1,s_2,s_3\}$ be the set chosen by the adversary. 

With probability $p_1$, we see $\os_1 \in \{s_2\pm 1,s_2\}, \os_2 \in \{s_3 \pm 1,s_3\}$ and win with probability $\frac{1}{2} \left( 1 - \frac{\log_2 (\os_2 - \os_1)}{\log_2 N} \right) \ge \frac{1}{2} \left( 1 - \frac{\log_2 (s_3 - s_2 + 2)}{\log_2 N} \right)$. Similarly, with probability $p_3$ when the largest number is deleted, we win with probability at least $\frac{1}{2} \left( 1 - \frac{\log_2 (s_2 - s_1 + 2)}{\log_2 N} \right)$.
 
With probability $p_2$, we see $\os_1 \in \{s_1 \pm 1,s_1\}, \os_2 \in \{s_3\pm 1,s_3\} $ and win with probability $\frac{\log_2 (\os_2 - \os_1)}{\log_2 N} \ge \frac{\log_2 (s_3 - s_1 - 2)}{\log_2 N}$.

 Therefore, the expected reward $\alg$ of the algorithm is at least 
\begin{multline*}
	\alg \ge p_1 \cdot \frac{1}{2} \left( 1 - \frac{\log_2 (s_3 - s_2 + 2)}{\log_2 N} \right) \cdot \frac{1}{p_1} + p_2 \cdot \frac{\log_2 (s_3 - s_1 - 2)}{\log_2 N} \cdot \frac{1}{p_2} + p_3 \cdot \frac{1}{2} \left( 1 - \frac{\log_2 (s_2 - s_1 + 2)}{\log_2 N} \right) \cdot \frac{1}{p_3} \\
	= 1 + \frac{\log_2 (s_3 - s_1 - 2)}{\log_2 N} -  \frac{\log_2 (s_3 - s_2 + 2) + \log_2 (s_2 - s_1 + 2)}{2\log_2 N} \\
	\ge 1 + \frac{\log_2 (s_3 - s_1 - 2)}{\log_2 N} -  \frac{\log_2 \left( \frac{s_3 - s_1 + 4}{2} \right)}{\log_2 N} = 1 + \frac{\log_2 \left(\frac{2(s_3-s_1-2)}{s_3-s_1+4} \right)}{\log_2 N} \ge 1+ \Omega\left(\frac{1}{\log N}\right),
\end{multline*}
where we use Jensen's inequality $\log a + \log b \le 2 \log \left( \frac{a+b}{2}\right)$ for the concave $\log(x)$ function in the second inequality, and in the last inequality, we know that $s_3 - s_1 \ge 40$ according to our technical assumption.

\paragraph{General Guessing Algorithm.}
We prove the theorem by induction and construct the algorithm recursively for each $n \ge 4$ using the algorithm for $n-1$. 
The input to our algorithm is an increasing sequence $\oset = (\os_i)_{i \in [n-1]}$. We observe the gaps $g_{i} \eqdef \os_{i+1} - \os_i$ for all $i \in [n-2]$.

We first introduce a strategy called \monogaps, that has expected reward significantly higher than $1$ when the sequence of gaps $\seqd \eqdef (d_i = s_{i+1} - s_i)_{i\in[n-1]}$ is not monotone, and does as well as random guessing for any instance.
Recall a technical assumption that every $d_i\ge 20$.
\begin{algorithm}[H]
	\caption{$\monogaps(n,\oset)$}
	\label{alg:monogaps}
	\SetKwFunction{FMain}{Recursive}
	\SetKwFunction{Fmg}{\monogaps}
	\SetKwFunction{Feg}{\expgaps}
	\SetKwProg{Fn}{Function}{:}{}
		Select a pair of two adjacent gaps $(g_i,g_{i+1})$, with $i \sim \uni[n-3]$\\
		\If {$g_i + 4 < g_{i + 1}$}{ 
			With probability $\frac{2}{3}$, \KwRet $i+2$ \\
			With probability $\frac{1}{3}$, \KwRet $i$
		}
		\If{$g_i > g_{i + 1} + 4$}{
			With probability $\frac{2}{3}$, \KwRet $i+1$ \\
			With probability $\frac{1}{3}$, \KwRet $i+3$
		}
		\KwRet $j \sim \text{Uni}\{i,i+3\}$
\end{algorithm}

\begin{lemma}
\label{lem:monogaps}
For any $S$, the expected reward $\mg$ of $\monogaps$ satisfies the following: 
	\begin{enumerate}
	\item $\mg \ge 1$;
	\item If $\seqd$ is not monotone, then $\mg \ge 1 + \frac{1}{3(n-3)}$;
	\item If $\seqd$ is increasing (or decreasing) and there exists $i$ with $d_i + d_{i+1} > d_{i+2} +8 $ (or $d_i +8 <  d_{i+1} + d_{i+2}$), then $\mg \ge 1 + \frac{1}{3(n-3)}$.
	\end{enumerate}
\end{lemma}

\paragraph{Intuition behind \monogaps\ algorithm.} It is useful to think about the random selection of the pair $(g_i,g_{i+1})$
as first guessing the deleted element to be among $\{i, i+1, i+2, i+3\}$. If this guess is correct, our decision only depends on 
$(g_i,g_{i+1})$, i.e., we reduce the problem to the case $n=4$ for $S=(s_i,s_{i+1},s_{i+2},s_{i+3})$ and if the corresponding part ($d_i,d_{i+1},d_{i+2}$) of $\seqd$ is not monotone we get a certain advantage over the random guessing strategy. For $n=4$ there are three following cases
\begin{enumerate}
\item $d_1\le d_2 \ge d_3$. In this case, when nature deletes $i=2$, we observe $(g_1=d_1+d_2>g_2=d_3)$ and when nature deletes 
$i=3$, we observe $(g_1=d_1<g_2=d_2+d_3)$. By guessing 
the deleted element $s_i$ to be inside of the larger gap $(g_1, g_2)$ we identify correctly the case $i=2$ and $i=3$ with probability $\frac{2}{3}$. This gives higher expected reward than $1$.  
\item $d_1\ge d_2 \le d_3$. In this case, the natural strategy of guessing-inside-the-larger-gap $(g_1,g_2)$ does not give us any advantage over the random guessing (but, it does not give us any disadvantage over random guessing). On the other hand, the strategy of guessing $i=1$ when $g_1<g_2$ and $i=4$ when $g_1>g_2$ is correct when $i=1$ or $i=4$. This allows us to improve upon random guessing strategy when we mix the guess-inside-the-large-gap and the guess-outside-in-the-direction-of-smaller-gap strategies.
\item In the case $d_1<d_2<d_3$ or $d_1>d_2>d_3$, either of the strategies gives expected reward of $1$.
\end{enumerate}
We note that the actual algorithm and the formal proof of Lemma~\ref{lem:monogaps} are more complicated than the above intuition, as the adversary can perturb a little the observed set $\oset$ and since the reduction to the case $n=4$ is only an informal statement. 

\noindent\begin{proof}
With probability $p_i$, $s_i$ is deleted and we observe $\oset$, where $\os_j \in \{s_j \pm 1, s_j\}$ for each $j \le i-1$, and $\os_j \in \{s_{j+1} \pm 1, s_{j+1}\}$ for each $j \ge i$.  We consider the sequence of distances $\seqd = (d_i = s_{i+1} - s_i)_{i\in[n-1]}$ for the original instance $S$.

\[
\os_{i-3} \underbrace{}_{g_{i-3}} \os_{i-2} \underbrace{}_{g_{i-2}} \os_{i-1} \underbrace{\qquad s_i \qquad}_{g_{i-1}} \os_{i}\underbrace{}_{g_{i}} \os_{i+1}\underbrace{}_{g_{i+1}} \os_{i+2}
\]

\[
s_{i-3} \underbrace{}_{d_{i-3}} s_{i-2} \underbrace{}_{d_{i-2}} s_{i-1} \underbrace{\qquad s_i \qquad}_{d_{i-1}+d_i} s_{i+1}\underbrace{}_{d_{i+1}} s_{i+2}\underbrace{}_{d_{i+2}} s_{i+3}
\]

We have $|g_{j}-d_{j}|\le 2$ for each $j\ne i-1$ and $|g_{i-1}-d_{i-1}-d_i|\le 2$. We first calculate the expected reward of algorithm~\ref{alg:monogaps} if $s_i$ was deleted from $S$. There are at most four possibilities for the random pair $(g_j,g_{j+1})$ that can lead to the correct guessing of $i$. Namely, $j\in\{i-3,i-2,i-1,i\}$:
\begin{enumerate}[(a)]
\item Algorithm selects the pair $(g_{i-3},g_{i-2})$. If $g_{i-3}>g_{i-2}+4$, then our expected reward is $\frac{1}{3}$. If $g_{i-2}-4 \le g_{i-3}\le g_{i-2}+4$, then our expected reward is $\frac{1}{2}$. 
Thus when $d_{i-3}\ge d_{i-2}$, we have $g_{i-3}\ge g_{i-2}- 4$ and algorithm's expected reward is at least $ \frac{1}{3}$.
\item Algorithm selects the pair $(g_{i-2},g_{i-1})$. If $g_{i-2}<g_{i-1}-4$ (happens when $d_{i-2} < d_{i-1} + d_i - 8$), then our expected reward is $\frac{2}{3}$. 
\item Algorithm selects the pair $(g_{i-1},g_{i})$. If $g_{i-1} > g_{i} + 4$ (happens when $d_{i-1} + d_i - 8 > d_{i+1} $), then our expected reward is $\frac{2}{3}$. 
\item Algorithm selects the pair $(g_{i},g_{i+1})$. If $g_{i} + 4 < g_{i+1}$, then our expected reward is $\frac{1}{3}$. If $g_{i}-4 \le g_{i+1}\le g_{i}+4$, then our expected reward is $\frac{1}{2}$. 
Thus when $d_{i+2}\ge d_{i+1}$, we have $g_{i+2}\ge g_{i+1}- 4$ and algorithm's expected reward is at least $\frac{1}{3}$.
\end{enumerate}
Therefore, the expected reward $\mg$ of the algorithm for any adversarial choice of $\oset$ is at least
\begin{multline}
\label{eq:expected_reward_minus_i}
\Ex{\mg(\oset) \cdot \ind{s_i \text{ is deleted}}}
\ge \frac{\ind{3 \le i \le n-1}}{n-3} \cdot \left(\frac{2}{3} \cdot  \ind{d_{i-1} + d_i> d_{i-2} + 8}  \right) \\
+ \frac{\ind{2 \le i \le n-2}}{n-3} \cdot \left( \frac{2}{3} \cdot \ind{d_{i-1} + d_{i} > d_{i + 1} + 8} \right)
+  \frac{\ind{4 \le i \le n}}{n-3}\cdot  \left(\frac{1}{3} \cdot \ind{d_{i-3}  \ge d_{i - 2}}\right)
\\
 + \frac{\ind{1 \le i \le n-3}}{n-3} \cdot \left(\frac{1}{3}  \cdot  \ind{d_{i+2} \ge d_{i+1}} \right),
\end{multline}
where the randomness in expectation is over the randomness of \monogaps. Hence,
\begin{multline}
\label{eq:expected_reward_monogaps}
	\mg(S) = \sum_{i \in [n]}  \Ex{\mg(\oset) \cdot \ind{s_i \text{ is deleted}}}\\
		\ge \frac{2}{3(n-3)}\cdot  \sum_{i = 2}^{n - 2} \left( \ind{d_i + d_{i + 1} > d_{i  - 1}+ 8} + \ind{d_{i - 1} + d_i > d_{i + 1}+ 8} \right)\\
		+\frac{1}{3(n-3)} \cdot \left(  \ind{d_1 \ge d_2} + \ind{d_{n - 1} \ge d_{n - 2}} + n-4 \right),
\end{multline}
where the inequality follows from~\eqref{eq:expected_reward_minus_i} and the fact that $\ind{d_i\ge d_{i+1}}+\ind{d_i\le d_{i+1}}\ge 1$.
Note that 
\be
\label{eq:monogaps_middle_indicators}
\ind{d_i + d_{i + 1} > d_{i  - 1}+ 8} + \ind{d_{i - 1} + d_i > d_{i + 1}+ 8}\ge 1 \text{ for every }i\in[n-1].
\ee
Indeed, if both of the indicators are $0$, we would have $0<d_i-8\le d_{i+1}-d_{i-1} \le 8-d_i<0$ (recall that $d_i\ge 12$), a contradiction. 

Now, if we estimate every $\ind{d_i + d_{i + 1} > d_{i  - 1}+ 8} + \ind{d_{i - 1} + d_i > d_{i + 1}+ 8}$ term by $1$ and ignore the terms $\ind{d_1 \ge d_2}$, $\ind{d_{n - 1} \ge d_{n - 2}}$, then the right hand side of~\eqref{eq:expected_reward_monogaps}
is at least 
\[
\mg(S)\ge \frac{2(n-3)+(n-4)}{3(n-3)}= 1-\frac{1}{3(n-3)}.
\] 

As it turns out, we can slightly improve this bound. First, observe that if there is an index $i$ such that both indicators $\ind{d_i + d_{i + 1} > d_{i  - 1}+ 8} = \ind{d_{i - 1} + d_i > d_{i + 1}+ 8}=1$, then $\mg(S)\ge 1-\frac{1}{3(n-3)}+\frac{2}{3(n-3)}=1+\frac{1}{3(n-3)}$. The latter immediately implies all three statements of the lemma. Therefore, it suffices to consider the case when all inequalities~\eqref{eq:monogaps_middle_indicators} are tight.

Second, we observe that if $d_{i-1}\le d_i\ge d_{i+1}$ for any $2\le i\le n-1$, then $\ind{d_i + d_{i + 1} > d_{i  - 1}+ 8} = \ind{d_{i - 1} + d_i > d_{i + 1}+ 8}=1$, which we assumed to be impossible. I.e., the sequence $\seqd$ does not have any internal ($1<i<n-1$) local maximums. This means that the sequence $\seqd$ is either 
\begin{enumerate}
\item strictly increasing, then $\ind{d_{n-1}\ge d_{n-2}}=1$ and  $\mg(S)\ge 1$; 
\item or strictly decreasing, then $\ind{d_{1}\le d_{2}}=1$ and  $\mg(S)\ge 1$; 
\item or strictly decreasing and then strictly increasing, then $\ind{d_{n-1}\ge d_{n-2}}=\ind{d_{1}\le d_{2}}=1$ and  
$\mg(S)\ge  1+\frac{1}{3(n-3)}$. This implies all three statements of the Lemma.
\end{enumerate}
To conclude the proof, we note that $\mg(S)\ge 1$, which implies the first statement of the Lemma~\ref{lem:monogaps}. Moreover, we have $\mg(S)\ge 1+\frac{1}{3(n-3)}$, unless $\seqd$ is a strictly monotone sequence, which implies the second statement of the Lemma~\ref{lem:monogaps}. Finally, if $d_i + d_{i+1} > d_{i+2} +8$ (or $d_i +8 <  d_{i+1} + d_{i+2}$) and $\seqd$ is strictly increasing (decreasing) sequence, i.e., $d_i<d_{i+1}<d_{i+2}$ (or $d_i>d_{i+1}>d_{i+2}$), then $\ind{d_i + d_{i + 1} > d_{i  - 1}+ 8} = \ind{d_{i - 1} + d_i > d_{i + 1}+ 8}=1$ and $\mg(S)\ge 1+\frac{1}{3(n-3)}$, which concludes the proof of the third part of Lemma~\ref{lem:monogaps}.
\end{proof}

Now, if we use \monogaps \ strategy with probability $1-O\left(\frac{1}{n}\right)$ we can ensure that sequence $\seqd$ has a nice structure, i.e., $\seqd$ is monotone and does not satisfy the Fibonacci-like inequality in the third point of Lemma~\ref{lem:monogaps}. Indeed, if the expected reward of \monogaps\ is at least $1+\Omega\left(\frac{1}{n}\right)$ we already have the expected reward to be higher than that of the random guessing regardless of what other strategy we use with probability $O\left(\frac{1}{n}\right)$, otherwise $\seqd$ has a nice structure and \monogaps\ gives us at least as good a reward as the random guessing.

In the following, we want to amplify the Fibonacci-like guarantee from Lemma~\ref{lem:monogaps} to much stronger condition that $d_{i+1}\ge C\cdot d_{i}$ for a large constant $C$, every $i\in[n-1]$  (analogously $d_i\ge C\cdot d_{i+1}$ for the decreasing $\seqd$). To do this, we introduce our next strategy $\expgaps(n,\oset)$. This strategy has an additional parameter $\ell\in [6]$ which we call a level of \expgaps. We are going to run \expgaps\ at every level $\ell$, with diminishing in $\ell$ probability.

\begin{algorithm}[H]
	\caption{$\expgaps(\ell,n,\oset)$}
	\label{alg:expgaps}
	Set level constants at $L_1 = 2 ,L_2  = 4, L_3 = 16, L_4 = 225, L_5=42374, L_6=2^{21}$\\
	Let $ I = \{1,2,n\}$ \quad \quad (or $I = \{1,n-1,n\}$ when $\seqg$ is decreasing) \\
	\If{$\seqg$ is increasing (or decreasing)}{
			$I \leftarrow I \cup \{i ~\mid~ n-1\ge i\ge 3,~~ g_{i-1} \ge L_{\ell} \cdot g_{i-2} + 2\cdot L_{\ell} + 2\}$\\ 
			 (or $I \leftarrow I \cup \{i ~\mid~ n-2\ge i\ge 2,~~ g_{i-1} \ge L_{\ell} \cdot g_i+ 2\cdot L_{\ell} + 2\}$ when $\seqg$ is decreasing)
		}
	\KwRet $i \sim \text{Uni}(I)$
\end{algorithm}

We say that an increasing (decreasing) $\seqd$ satisfies level-$\ell$ condition for $\ell\in [6]$ if and only if $d_{i+1}\ge L_{\ell}\cdot d_i$ 
($d_{i}\ge L_{\ell}\cdot d_{i+1}$) for every $i\in[n-2]$. We also introduce the level-$0$ condition which just refers to the
Fibonacci-like condition from Lemma~\ref{lem:monogaps}:
\[
\forall i\in[n-3]~~ d_i + d_{i+1} \le d_{i+2} +8  \quad\quad (\text{or } d_i +8 \ge  d_{i+1} + d_{i+2} \text{ for decreasing }\seqd)
\]
\begin{lemma}
\label{lem:expgaps}
If an increasing sequence $\seqd$ satisfies level-$(\ell-1)$ condition, then 
the expected reward $\eg(\ell)$ of $\expgaps(\ell,n,\oset)$ is at least
	\begin{enumerate}	
		\item $\eg(\ell) \ge 1$;
		\item If $\seqd$ violates level-$\ell$ condition, then $\eg(\ell) \ge 1 + \frac{n-3}{n(n-1)}$.	
	\end{enumerate}
\end{lemma}
\begin{proof} We first observe the following useful property of the set $I$ in the \expgaps\ strategy.
\begin{claim}
\label{cl:when_I_nonempty_expogaps}
$\forall i \in [n]$, when $s_i$ is deleted from $S$, then $i \in I(\oset)$ for any adversarial choice of $\oset$.
\end{claim}
\begin{proof} 
The case when $i \in \{1,2,n\}$ is trivial according to our algorithm. We consider the case when $s_i$  is deleted from $S$ for a given $i\notin\{1,2,n\}$. Then $g_{i-2}=\os_{i-1}-\os_{i-2}$ and $g_{i-1}=\os_{i+1}-\os_{i-1}$ satisfy 
$|g_{i-2}-d_{i-2}|\le 2$ and $|g_{i-1}-d_{i-1}-d_i|\le 2$. 
We consider two cases.
\begin{enumerate}
\item When $\ell=1$. The level-0 condition implies $d_{i-1}+d_{i-2} \le d_i + 8$. Then we have,
\[
g_{i-1}\ge d_{i-1} + d_{i} - 2 \ge 2d_{i-1} + d_{i-2} -10 \ge 2d_{i-2} + 10 \ge 2g_{i-2}+6, \\
\]
where the third inequality follows from the fact that $d_{i-1} \ge d_{i-2} \ge 20$. 
\item When $\ell > 1$. The level-$(\ell-1)$ condition implies $L^2_{\ell-1}\cdot d_{i-2} \le L_{\ell-1}\cdot d_{i-1} \le d_{i}$. Hence
\begin{multline*}
g_{i-1} \ge d_{i-1} + d_i -2 \ge (L^2_{\ell-1} + L_{\ell-1})\cdot d_{i-2} -2 \\
\ge L_{\ell} \cdot \left(d_{i-2}+2 \right) + 2 L_{\ell} +2 \ge L_{\ell} \cdot g_{i-2} + 2L_{\ell} +2.
\end{multline*}
Here, the third inequality holds since  
$
L_{\ell} \le \frac{(L^2_{\ell-1} + L_{\ell-1})\cdot 20 -4}{20+4} \le \frac{(L^2_{\ell-1} + L_{\ell-1})\cdot d_{i-2} -4}{d_{i-2}+4},
$
and according to the choice of the constants $L = \left<2,4,16,225,42374, 2^{21}\right>$. 
\end{enumerate}
In both cases, the algorithm adds $i$ to $I(\oset)$, since $g_{i-1}\ge L_{\ell} \cdot g_{i-2} + 2L_{\ell} +2$.
\end{proof}

Now, we prove the first statement of the lemma. Note that when $s_i$ is deleted from $S$ for any $i\in[n]$, then $i\in I(\oset)$ and the expected reward of $\expgaps(\ell,n,\oset)$ is at least $\frac{1}{|I(\oset)|} \ge \frac{1}{n}$.
\[
\eg(S) = \sum_{i \in [n]}  \Ex{\eg(\oset) \cdot \ind{s_i \text{ is deleted}}} \ge \sum_{i \in[n]} \frac{1}{n} = 1.
\]
Next, if $\seqd$ violates level-$\ell$ condition then there are cases when $I(\oset)$ has less than $n$ elements.
\begin{claim}
\label{cl:when_II_nonempty_expogaps}
If $L_{\ell}\cdot d_{j-1} > d_j$ for $n>j>1$, then 
\begin{enumerate} 
\item for every deletion of $i\ge j+2$ and every $\oset$ we have $j+1\notin I(\oset)$;
\item for every deletion of $i\le j-2$ and every $\oset$ we have $j\notin I(\oset)$;
\end{enumerate}
\end{claim}
\begin{proof}
We only prove the first statement, as the second statement only differs by a shift of indexes.
We have
$
g_{j} \le d_j + 2 < L_{\ell} \cdot d_{j-1} + 2 \le L_{\ell} \cdot \left( g_{j-1}+2 \right) + 2 = L_{\ell} \cdot g_{j-1}+2 L_{\ell} + 2,
$
since $g_j \le d_j + 2$ and $d_{j-1} \le g_{j-1}+2$. I.e., $j+1\notin I(\oset)$ when $i\ge j+2$ was deleted.
\end{proof}

Now, suppose there exists a $2\le j \le n-1$ with $L_{\ell}  \cdot d_{j-1} > d_j$. Then for every $i \in P \eqdef [n]\setminus\{j-1,j,j+1\}$,
when $s_i$ is deleted from $S$, the corresponding $I(\oset)$ has size at most $n-1$, which results in an improved performance of our algorithm. Namely,
\[
\eg(S) = \sum_{i \in [n]}  \Ex{\eg(\oset) \cdot \ind{s_i \text{ is deleted}}} \ge \sum_{i \in P} \frac{1}{n-1} + \sum_{i \notin P} \frac{1}{n} \ge 1 + \frac{n-3}{n^2-n}. \qedhere
\]
\end{proof}

Finally, we present our recursive guessing algorithm.

\begin{algorithm}[H]
	\caption{\textbf{Guess}$(n,\oset)$}
	\label{alg:recursive}
	\SetKwFunction{FMain}{Recursive}
	\SetKwFunction{Fmg}{\monogaps}
	\SetKwFunction{Feg}{\expgaps}
	\SetKwProg{Fn}{Function}{:}{}
	\eIf{$n=3$}{Run the algorithm in the warm-up.}{
		With probability $1-\frac{1}{6n}$, \KwRet $\monogaps(n,\oset)$; \\
		For each $\ell \in [6]$, with probability $\frac{1}{(6n)^\ell}-\frac{1}{(6n)^{\ell+1}}$, \KwRet $\expgaps(\ell,k,\oset)$; \\
		With remaining probability $\frac{1}{(6n)^{7}}$, let $\otset = \left( \ot_i \eqdef \lfloor \log_2 g_i \rfloor \right)_{i \in [n-2]}$. \\
			\eIf{$(t_i)_{i \in [n-2]}$ is increasing (or decreasing)}{
				\KwRet $\textbf{Guess}(n-1, \otset) + 1$ (or $n - \textbf{Guess}(n-1,\textbf{Reverse}(\otset))$\footnote{The \textbf{Reverse} function reverses the descending vector $\oset$ to ascending.})
			}{
				\KwRet $i \sim \text{Uni}\{1,2,\dots,n\}$
			}
		}
\end{algorithm}

\begin{lemma} 
If an increasing $\seqd$ sequence violates level-$6$ condition, the expected reward $\textsf{Guess}$ of our algorithm is at least $1+ \Omega\left( \frac{1}{n^{7}} \right)$.
\end{lemma}
\begin{proof}
By Lemma~\ref{lem:monogaps}, when $\seqd$ violates the level-$0$ condition, we have
\[
\textsf{Guess} \ge \left(1-\frac{1}{6n} \right) \mg \ge \left(1-\frac{1}{6n} \right) \cdot \left( 1 + \frac{1}{3(n-3)} \right) \ge 1 + \Omega\left(\frac{1}{n}\right).
\]
Otherwise, suppose $\seqd$ satisfies level-$(\ell-1)$ condition while violates level-$\ell$ condition. It must also satisfy level-$j$ conditions for all $j \le \ell-2$. By Lemma~\ref{lem:monogaps} and \ref{lem:expgaps}, $\monogaps$ and every $\expgaps$ with level at most $\ell-1$ give an expected reward of $1$. Moreover, $\expgaps(\ell)$ achieves an expected reward of $1+\frac{n-3}{n(n-1)}$. Therefore, the expected gain of our algorithm is
\begin{multline*}
\textsf{Guess} \ge \left(1- \frac{1}{(6n)^{\ell}} \right) \cdot 1+ \left(\frac{1}{(6n)^{\ell}} - \frac{1}{(6n)^{\ell+1}} \right) \eg(\ell) \\
\ge \left(1- \frac{1}{(6n)^{\ell}} \right) + \left(\frac{1}{(6n)^{\ell}} - \frac{1}{(6n)^{\ell+1}} \right) \cdot \left( 1+ \frac{n-3}{n(n-1)}\right) \ge 1+\Omega \left( \frac{1}{n^{\ell+1}} \right).\qedhere
\end{multline*}
\end{proof}

With the above lemma, when $\seqd$ violates the level-$6$ condition, the expected reward of our algorithm is $1+\Omega \left(\frac{1}{n^7}\right)$, which is better than the stated bound of Theorem~\ref{thm:guessing}. 

In the remainder of the proof, we focus on the case when $\seqd$ satisfies level-$6$ condition. Without loss of generality, we consider the case when $\seqd$ is increasing and $d_{i} \ge 2^{21} \cdot d_{i-1}$ for every $i\ge 2$.

We construct an instance $T = \{ t_i \eqdef \lfloor \log_2d_i \rfloor \}_{i \in [n-1]}$ of $(n-1)$ numbers. Note that the largest number of $T$ is at most $\log N$. Let each $i$ be deleted with probability 
$q_i \eqdef \begin{cases}
		p_1 + p_2 & i = 1 \\
		p_{i + 1} & i \ge 2
\end{cases}$.

First of all, we verify that the instance satisfies the technical assumption.
\begin{claim}
For every $2 \le i \le n-1$, we have $t_{i} - t_{i-1} \ge 20$.
\end{claim}
\begin{proof}
$t_i - t_{i-1} = \lfloor \log_2d_i \rfloor - \lfloor \log_2d_{i-1} \rfloor \ge \log_2d_i -\log_2d_{i-1} - 1 \ge \log_2 L_{6}-1 \ge 20$.
\end{proof}

This claim is the reason why we needed to apply multiple levels of \expgaps.

Next, we construct a correspondence between the perturbed guessing game on $T$ and the recursive part of the algorithm, where we guess $\textbf{Guess}(n-1, \otset) + 1$. Consider the following two cases:

\begin{itemize}
\item The case when $i=1,2$ is deleted from $S$, which happens with probability $p_1+p_2=q_1$, corresponds to the case when $1$ is deleted from $T$. We verify that after the deletion of $1$ from $T$, all other numbers are perturbed by at most $1$.
\begin{itemize}
\item When $i=1$, we have $|g_i-d_{i+1}| \le 2$ for every $i \in [n-2]$. Thus 
\begin{multline*}
\left| \ot_i - t_{i+1} \right| = \left| \lfloor \log_2 g_i \rfloor - \lfloor \log_2 d_{i+1} \rfloor\right| \le \left| \lfloor \log_2 (d_{i+1}+2) \rfloor - \lfloor \log_2 d_{i+1} \rfloor\right| \\
\le \left| \lfloor \log_2 (2 \cdot d_{i+1}) \rfloor - \lfloor \log_2 d_{i+1} \rfloor\right| = 1~.
\end{multline*}
\item When $i=2$, we have $|g_1-d_{1}-d_2| \le 2$ and $|g_i-d_{i+1}| \le 2$ for every $2\le i \le n-2$. Thus,
\begin{multline*}
\left| \ot_1 - t_{2} \right| = \left| \lfloor \log_2 g_1 \rfloor - \lfloor \log_2 d_2 \rfloor\right| \le \left| \lfloor \log_2 (d_1+d_2+2) \rfloor - \lfloor \log_2 d_{2} \rfloor\right| \\
\le \lfloor \log_2 (2\cdot d_2) \rfloor - \lfloor \log_2 d_{2} \rfloor = 1~.
\end{multline*}
The difference between $\ot_i$ and $t_{i+1}$ for $i \ge 2$ is the same as the first case.
\end{itemize}
\item The case when $i>2$ is deleted from $S$, which happens with probability $p_i=q_{i-1}$, corresponds to the case when $i-1$ is deleted from $T$. Again, we verify that after the deletion of $1$ from $T$, all other numbers are perturbed by at most $1$.
Observe that in this case, $|g_j - d_j| \le 2$ for $j \le i-2$; $|g_{i-1} - d_{i-1}-d_i | \le 2$; and $|g_j - d_{j+1}| \le 2$ for $j \ge i$.
\begin{itemize}
\item For $j \le i-2$, 
$\left| \ot_j - t_{j} \right| = \left| \lfloor \log_2 g_j \rfloor - \lfloor \log_2 d_{j} \rfloor\right| \le \left| \lfloor \log_2 (d_{j}+2) \rfloor - \lfloor \log_2 d_{j} \rfloor \right|\le 1~.$
\item For $j = i-1$, 
\begin{multline*}
\left| \ot_{i-1} - t_{i} \right| = \left| \lfloor \log_2 g_{i-1} \rfloor - \lfloor \log_2 d_{i-1} \rfloor\right| \le \left| \lfloor \log_2 (d_{i-1}+d_i+2) \rfloor - \lfloor \log_2 d_{i} \rfloor \right| \\
\le \left| \lfloor \log_2 (2\cdot d_i) \rfloor - \lfloor \log_2 d_{j} \rfloor \right| \le 1~.
\end{multline*}
\item For $j \ge i$, 
$\left| \ot_j - t_{j+1} \right| = \left| \lfloor \log_2 g_j \rfloor - \lfloor \log_2 d_{j+1} \rfloor\right| \le \left| \lfloor \log_2 (d_{j}+2) \rfloor - \lfloor \log_2 d_{j+1} \rfloor \right|\le 1~.$
\end{itemize}
\end{itemize}

When $s_i$ is deleted from $S$ for $i=1,2$, it corresponds to the same deletion of $t_1$ from $T$. According to our algorithm, we will consistently guess $2$ to $S$ if the recursive algorithm makes a guess of $1$ to $T$. Though the probability of guessing correctly will be $p_2 \le q_1 = p_1+p_2$, the expected reward will be scaled up proportionally. When $s_i$ is deleted from $S$ for $i>2$, the rewards in both games are the same.
 
Therefore, the expected reward of the recursive part of our algorithm equals:
\begin{multline*}
\textsf{Recursive}(S)
= \textsf{Guess}(T) \ge 1 + \frac{1}{(6(n-1))^{7(n-1)}} \cdot \Omega\left( \frac{1}{\klog[n-3](\log N)} \right) \\
= 1 + \frac{1}{(6n)^{7(n-1)}} \cdot \Omega\left( \frac{1}{\klog[n-2]N} \right),
\end{multline*}
where the inequality follows from the induction hypothesis and that the largest number in $T$ is at most $\log N$. 
Finally, by Lemma~\ref{lem:monogaps} and \ref{lem:expgaps}, we have that $\monogaps$ and $\expgaps$ of all levels have expected reward at least $1$ when $\seqd$ satisfies level-$\ell$ condition. With a constant probability of $\frac{1}{(6n)^7}$ executing the recursive step, we achieve an expected reward of 
\[
1 \cdot \left( 1-\frac{1}{(6n)^7} \right) + \left(1 + \frac{1}{(6n)^{7(n-1)}} \cdot \Omega\left( \frac{1}{\klog[n-2]N} \right) \right)\cdot \frac{1}{(6n)^7} \ge  1 + \frac{1}{(6n)^{7n}} \cdot \Omega\left( \frac{1}{\klog[n-2]N} \right).
\] 
\section{Proof of Theorem~\ref{thm:googol}}
\label{app:googol}
We use slightly different notations for the gaps. Namely, the sequence of visible numbers $(s_1< s_2 <\ldots <s_k)$ at step $k$ has a vector of gaps $(d_1,\ldots,d_k)$ where $d_1=s_1, d_{i+1}=s_{i+1}-s_i$ for $i\in[k-1]$. Recall that in the construction $d_i=s_{i}-s_{i-1}\sim L_{\levi}$ for a permutation of levels 
$\lev\sim\distlev$ with $\levi[1]=n$ and that $\{L_i = \uni [\Delta^i]\}_{i=1}^{n}$ for $\Delta=\frac{n}{\eps}$.  We assume to the contrary that there is an online algorithm $\alcard$ in the cardinal setting that is significantly better in expectation over $\vect{d}\sim\distset$ than the ordinal algorithm $\alord$  in the game of googol. I.e., $\Ex[\distset,\sym(n)]{\alcard(\vect{d},\pi)}>\Ex[\sym(n)]{\alord(\pi)}+\eps$. We simulate $\alcard$ in the secretary level setting and show 
that its simulation does not work only on insignificant fraction of inputs. This leads to a contradiction with the fact that $\allev$ is no better than the ordinal algorithm $\alord$.

Given the gaps levels $\lev=(\levi[1],\ldots,\levi[n])$, we can easily apply our construction with $\{d_i\sim L_{\levi}\}_{i=1}^{n}$ and  run $\alcard$ on that simulated instance. It is straightforward to check that for randomly generated gaps $d_i\sim L_{\levi}$ the combined gap $d_i+d_{i+1}$ has a very similar distribution to $L_{\max(\levi,\levi[i+1])}$. The main challenge is that we have to construct the gaps \emph{online} and make sure that they are consistent throughout all $n$ steps. Specifically, to use $\alcard$ in the level setting, we need to specify the sequence of gaps $\dgapv^k=(\dgapi[1]^k,\dgapi[2]^k,\ldots,\dgapi[k]^k)$ at each step $k\in[n]$ from a sequence of levels $\lev^k=(\lev^k_1,\ldots,\lev^k_k)$ and a visible relative ranking $\pi^k$ among the first $k$ numbers. The simulation $\simul$ works as follows.

\begin{enumerate}
	\item For each level $i\in[n]$, sample $r_i\sim L_i$. 
	\item If $r_i\le \sum_{j<i}r_j$ for any $i\in [n]$ reject instance (i.e., simulation has failed).
	\item Set $\dgapi[1]^1=r_n$ in the first step. Take or skip the $1$-st element, same as $\alcard(\dgapi[1]^1)$.
	\item For each step $k+1$ for $k\in[n-1]$,
	  \begin{itemize}
		    \item Observe (from $\pi^k$ and $\lev^k$) the $j$-th gap $\dgapi[j]^k$ where the new $k+1$-th element 
							arrives;
				\item Observe two new levels $\lev^{k+1}_j,\lev^{k+1}_{j+1}$ (with $\max\{\lev^{k+1}_j,\lev^{k+1}_{j+1}\}=\lev^k_j$). 
							Let $x$ be the index of the smaller level: $x=j$ if $\lev^{k+1}_j<\lev^{k+1}_{j+1}$, and $x=j+1$ if 
							$\lev^{k+1}_{j}>\lev^{k+1}_{j+1}$; and $y$ be the index of the larger level: $y=j+(j+1)-x$  
				\item Set new gaps $(\dgapi[1]^{k+1},\dgapi[2]^{k+1},\ldots,\dgapi[k+1]^{k+1})$ as: 
							\[
								\dgapi^{k+1}\eqdef
								\begin{cases}
									\dgapi^{k}, & \text{for all }i< j,\\
									r_{\lev^{k+1}(x)}, & \text{for } i=x\\
									\dgapi[j]^k-\dgapi[x]^{k+1}, & \text{for } i=y\\
									\dgapi[i-1]^k, & \text{for } i>j+1.
								\end{cases}
							\]				
				\item Take or skip $k+1$-th element, same as  $\alcard(\dgapv^{k+1},\pi^{k+1})$.    			
		\end{itemize}		
\end{enumerate}
Here are two of simple observations about simulation $\simul$.
\begin{claim}
	\label{cl:correct_simulation}
	If the simulation $\simul$ does not fail at the second step, then it produces an online sequence of gaps 
	consistent with a cardinal instance $\dgapv^n$. Moreover, the probability to obtain such
	instance $(\dgapv^n,\pi)$ in our simulation $\simul$ is the same $\frac{1}{n!}\prod_{k=1}^n \frac{1}{\Delta^k}$ as in 
	the distribution $\distset$.
\end{claim}
\begin{proof}
Observe that (i) every gap $\dgapi^k\le\Delta^{\levi^k}$ of $\levi^k$-th level the first time it appears in the sequence $\lev^k$ and (ii) it may only get smaller after step $k$. Moreover, we subtract only smaller levels of gaps from $\dgapi^k$ and not more than one time per each level. Thus our condition $r_i> \sum_{j<i}r_j$ ensures that $\dgapi^k>0$ at any time. Hence, if the simulation $\simul$ does not fail at the second step, the gaps $\dgapi^k\in\supp(L_{\levi^k})$ for all $i$ and $k$ and the sequence of gaps $\dgapv^{k+1}$ is consistent by the construction. 

To obtain the second part of the claim, observe that for any fixed arrival order $\pi$ the mapping from randomly generated $(r_1,\ldots,r_n)$ to $\dgapv^n$ is injective, the probability to see any given arrival order $\pi$ is $\frac{1}{n!}$, and $\prob{(r_1,\ldots,r_n)}=\prod_{k=1}^n \frac{1}{\Delta^k}$.
\end{proof}

\begin{claim}
	\label{cl:simulation_difference}	
For any arrival order $\pi\in\sym(n)$,	the probability that simulation $\simul$ fails is not more than $\eps$. Moreover,
the distribution of instances in the cardinal setting obtained in $\simul$ is close to the distribution $\distset$
\[
\dtv\left(\{\pi,\dgapv^n\sim\simul\},\{\pi,\vect{d}\sim\distset\}\right)\le\eps.
\]  	
\end{claim}
\begin{proof}
Fix any arrival order $\pi\in\sym(n)$ and feasible permutation of levels $\lev\in\sym(n)$ ($\levi[1]=n$). By the union bound 
\[
\Prx{\simul \text{ fails}}\le
\sum_{i\ge 2}^n \Prx{r_i\le \sum_{j<i}r_j}\le \sum_{i\ge 2}^n \Prx{r_i\le\frac{\Delta^i-1}{\Delta-1}}\le \frac{n-1}{\Delta-1}\le\frac{n}{\Delta}\le\eps,
\]
where the second inequality follows as each $r_j\le\Delta^j$ for $j\in[i-1]$; third inequality holds as $\Delta>n$; and forth inequality holds as $\Delta\ge\frac{n}{\eps}$. By Claim~\ref{cl:correct_simulation}, construction of $\distset$, and the above bound, we have for any fixed arrival order $\pi$ and feasible sequence of levels $\lev$ that
$\dtv\left(\dgapv^n(\pi,\lev,\vect{r}),\vect{d}\sim\distlev |\lev\right)\le\eps$ 
which concludes the proof of the claim.
\end{proof}
Finally, we arrive at a contradiction as follows 
\[
\Exlong[\pi]{\alord(\pi)}=\Exlong[\pi,\distlev]{\allev(\pi,\lev)}\ge\Exlong[\pi,\distlev]{\Exlong[\vect{r}(\lev)]{\simul(\pi,\vect{r})}}\\
\ge\Exlong[\pi,\distset]{\alcard(\pi,\vect{d})}-\eps > \Exlong[\pi]{\alord(\pi)},
\]
where the first equality holds by Theorem~\ref{thm:level_ordinal}; the first inequality holds as the best algorithm in the level setting is at least as good as  simulation algorithm $\simul$; the second inequality holds by Claim~\ref{cl:simulation_difference} and the fact that the reward in the game of googol is never more than $1$; the last inequality holds by the assumption that cardinal algorithm does significantly better than the ordinal algorithm.  

\end{document}